\pdfoutput=1
\RequirePackage{ifpdf}
\ifpdf 
\documentclass[pdftex]{sigma}
\else
\documentclass{sigma}
\fi

\numberwithin{equation}{section}

\newtheorem{Theorem}{Theorem}[section]
\newtheorem{Corollary}[Theorem]{Corollary}
\newtheorem{Lemma}[Theorem]{Lemma}
 { \theoremstyle{definition}
\newtheorem{Remark}[Theorem]{Remark} }

\begin{document}

\newcommand{\arXivNumber}{2004.00933}

\renewcommand{\PaperNumber}{135}

\FirstPageHeading

\ShortArticleName{Toward Classification of 2nd Order Superintegrable Systems}

\ArticleName{Toward Classification of 2nd Order Superintegrable\\ Systems in 3-Dimensional Conformally Flat Spaces\\ with Functionally Linearly Dependent Symmetry\\ Operators}

\Author{Bjorn K.~BERNTSON~$^\dag$, Ernest G.~KALNINS~$^\ddag$ and Willard MILLER Jr.~$^\S$}

\AuthorNameForHeading{B.K.~Berntson, E.G.~Kalnins and W.~Miller~Jr.}

\Address{$^\dag$~Department of Mathematics, KTH Royal Institute of Technology, Stockholm, Sweden}
\EmailD{\href{mailto:bbernts@kth.se}{bbernts@kth.se}}

\Address{$^\ddag$~Department of Mathematics, University of Waikato, Hamilton, New Zealand}
\EmailD{\href{mailto:math0236@waikato.ac.nz}{math0236@waikato.ac.nz}}

\Address{$^\S$~School of Mathematics, University of Minnesota, Minneapolis, Minnesota, USA}
\EmailD{\href{mailto:miller@ima.umn.edu}{miller@ima.umn.edu}}
\URLaddressD{\url{http://www-users.math.umn.edu/~mille003/}}

\ArticleDates{Received April 07, 2020, in final form December 09, 2020; Published online December 16, 2020}

\Abstract{We make significant progress toward the classification of 2nd order superintegrable systems on 3-dimensional conformally flat space that have functionally linearly dependent (FLD) symmetry generators, with special emphasis on complex Euclidean space. The symmetries for these systems are linearly dependent only when the coefficients are allowed to depend on the spatial coordinates. The Calogero--Moser system with 3~bodies on a~line and 2-parameter rational potential is the best known example of an FLD superintegrable system. We work out the structure theory for these FLD systems on 3D conformally flat space and show, for example, that they always admit a 1st order symmetry. A partial classification of FLD systems on complex 3D Euclidean space is given. This is part of a~project to classify all 3D 2nd order superintegrable systems on conformally flat spaces.}


\Classification{20C35; 35B06; 70H20; 81Q80; 81R12}

\section{Introduction}
There is a hierarchy of 2nd order classical and quantum superintegrable systems in 3-dimensional conformally flat spaces, ranging from the most tractable at the top, nondegenerate (i.e., 4-parameter) potentials with 6 linearly independent symmetries, all of which have been classified, followed by semidegenerate (i.e., 3-parameter) potentials on which much progress has been made, to the least tractable (1-parameter) for classification at the bottom. By definition the 2 classes at the top admit 5 functionally linearly independent symmetry operators, i.e., they are not only linearly independent in the usual sense but also if the coefficients are allowed to depend on the spatial variables. However, there exist 2nd order superintegrable systems with at least 5 functionally linearly dependent symmetry operators and 2-parameter potentials; such systems have never been classified. We initiate the study of such systems by developing their structure theory on conformally flat spaces and performing a partial classification of these systems in constant curvature spaces.

We recall some basic facts and results about conformally flat superintegrable systems.
An $n$-dimensional complex Riemannian space is conformally flat if and only if
it admits a set of local coordinates $\{x_1,\dots,x_n\}$ such that the
contravariant metric tensor takes the form
$g^{ij}=\delta^{ij}/\lambda({\bf x})$ \cite{KKM2005,KKM2018}. A classical
superintegrable system ${\mathcal H} =\sum_{ij}g^{ij}p_ip_j+V({\bf x})$ on
the phase space of this manifold is one that
admits $2n-1$ functionally independent generalized
symmetries (or constants of the motion) ${\mathcal S}_{k}$ for $k=1,\dots,2n-1$ with
${\mathcal S}_{1}={\mathcal H}$ where the ${\mathcal S}_{k}$ are polynomials in the momenta~$p_j$. It is easy to see that $2n-1$ is the maximum possible
number of functionally independent symmetries and, locally,
such (in general nonpolynomial) symmetries always exist. The system is second order maximal superintegrable if the $2n-1$ functionally independent symmetries can be chosen
to be quadratic in the momenta. (Second order superintegrable systems,
though complicated, are tractable because standard orthogonal separation of variables techniques are associated
with second order symmetries, and these techniques can be brought to bear.)

For a classical 3D system in a conformally flat space (note that all
2D spaces are conformally flat) we can
always choose local coordinates $\{x,y,z\}$, not unique, such that the
Hamiltonian takes the form
${\mathcal H}=\big(p_1^2+p_2^2+p_3^2\big)/{\lambda(x,y,z)}+V(x,y,z).$
This system is {\it second order superintegrable} with semidegenerate potential
$V=V(x,y,z;\alpha,\beta,\gamma)=\alpha V^{\alpha}({\bf x}) +\beta V^{\beta}({\bf x}) +\gamma V^{\gamma}({\bf x}) $ if it admits 5 functionally independent
quadratic constants of the motion, i.e., generalized symmetries,
\begin{equation*} 
{\mathcal S}_{k}=\sum_{i,j}a^{ij}_{k}p_ip_j+W_{k}(x,y,z;\alpha,\beta,\gamma)={\mathcal S}_{k}^{0}+W_{k},\qquad k=1,\dots,5.\end{equation*}
 Here the functions $V^{\alpha}$, $V^{\beta}$, $V^{\gamma}$ are independent of the parameters $\alpha$, $\beta$, $\gamma$, the set $\big\{V^{\alpha},V^{\beta},V^{\gamma}\big\}$ must have linearly independent gradients, and we ignore the additive constant. We call this a~3-parameter potential.

In some cases the system may also have a 6th symmetry ${\mathcal S}_{6}$, (but no more) such the set $\big\{{\mathcal S}_{k}^0 \,|\, k=1,\dots,6\big\}$ is functionally linearly independent and this implies that the potential depends on 4 parameters~\cite{KKM2005}. Furthermore the classification theory requires that the 5, 6 constants of the motion be {\it functionally linearly independent}, i.e., the equation
\begin{equation}\label{FLI} \sum_{k=1}^{5,6}f_{k}({\bf x})\mathcal{S}_k^0 ({\bf x})=0\end{equation}
is satisfied if and only if $f_{k}({\bf x})=0$ for all $k$. If equation (\ref{FLI}) is satisfied for functions $f_{k}({\bf x})$ not identically $0$, the set of constants of the motion are {\it functionally linearly dependent} (FLD).

For 2nd order superintegrable systems in 3 dimensions that are functionally linearly independent, the systems that admit 6 linearly independent second order constants of the motion (the maximum possible) have all been classified~\cite{CK, KKM2007b} and there has been considerable progress on the remaining 5 linear independent case~\cite{EM,KKM2007a}. However, little has been done to classify superintegrable systems in 3 dimensions that are FLD. The best known such system is the rational 3-body Calogero--Moser system on the line with 2-parameter potential. To the best of our knowledge there are no 2nd order FLD superintegrable systems with trigonometric, hyperbolic, or elliptic potentials. In this paper we derive structure results for all 2nd order superintegrable FLD systems with $r\ge 5$ linearly independent second order symmetries on conformally flat real or complex spaces that have potentials that depend on 2 functionally independent variables (the maximum possible), and such that the FLD equation $\sum_{k}f_{k}({\bf x})\mathcal{S}_{k}^0({\bf x})=0$ is satisfied with at most~5 nonzero terms $f_{k}({\bf x})$. (For the analogous 2nd order 2-dimensional FLD systems the answer is known: there is only one such family of systems~\cite{KKM2005c}.)

The paper is organized as follows: in Section~\ref{section2} we present the Calogero--Moser system and a system on 3-dimensional Minkowski space as examples. In Section~\ref{4} we present structure results for all FLD systems on conformally flat spaces. The most important result is that all such systems admit a 1st order constant of the motion. In Section~\ref{5} we work out a partial classification of all 3-dimensional second order superintegrable FLD systems in flat space with 2-parameter potentials, such that the FLD equation $\sum_{k}f_{k}({\bf x})\mathcal{S}_{k}^0({\bf x})=0$ is satisfied with at most 5 nonzero terms $f_{k}({\bf x})$, (including the structure of the symmetry algebras for most of these systems). In Section~\ref{6} we summarize the corresponding result for 3-dimensional FLD systems on the complex 3-sphere. In Section~\ref{7} we present some conclusions and a brief discussion of related properties of these systems. Here all of our systems are classical. However the quantum analogs follow easily by symmetrization of the symmetry operators and there is a 1-1 matching of the Hamiltonians modulo the scalar curvature~\cite{KKM2006}. In particular the Euclidean space Hamiltonians are identical.

\section{Examples} \label{section2}

\subsection[An FLD example: the rational Calogero--Moser system with 2-parameter potential]{An FLD example: the rational Calogero--Moser system\\ with 2-parameter potential}

This potential takes the form \cite{BCR,Cal1,Cal2,Evans,Moser,Ran,RWW,SW,WS}
\begin{equation}\label{Cal1}
V=\frac{\alpha}{(x-y)^2}+\frac{\beta}{(y-z)^2}+\frac{\gamma}{(z-x)^2}.
\end{equation}
Let us consider the system of symmetries defining the
system with potential $V$. A basis for the space of symmetries is
(using $J_{12}=xp_2-yp_1$, $J_{23}=yp_3-zp_2$, $J_{13}=xp_3-zp_1$)
\begin{gather*}
{\mathcal S}_{1}= {\mathcal H} = p_1^2+p_2^2+p_3^2+V, \qquad {\mathcal S}_{2}=(p_1+p_2+p_3)^2, \qquad
{\mathcal S}_{3}= J_{12}^2+J_{23}^2+J_{13}^2+W_3, \\
{\mathcal S}_{4}= p_1(J_{13}-J_{12})+p_2(J_{12}-J_{23})+p_3(J_{23}-J_{13})+W_4, \\
{\mathcal S}_{5}= J_{12}J_{13}+J_{23}J_{12}+J_{13}J_{23}+W_5,
\end{gather*}
where the potential terms $W_i$ contain the parameters $\alpha$, $\beta$, $\gamma$. In this case, the Bertrand--Darboux equations \cite{KKM2005,KKM2006} for each symmetry ${\mathcal S}_{k}=\sum_{ij} a^{ij}_{k} p_ip_j+W_k$ of ${\mathcal H}$ are
\begin{gather} V_x+V_y+V_z=0,\qquad
(x-y)V_{xy}+(z-y)V_{yz}-V_x+2V_y-V_z=0, \nonumber\\
 (x-z)V_{xz}+(y-z)V_{yz}-V_x-V_y+2V_z=0, \label{BD1}
\end{gather}
and their differential consequences.

We say that this is a (functionally independent) 2-parameter potential. A 2-parameter potential is one that can be expressed in the form $V =a_1 f(x,y,z)+a_2 g(x,y,z)$ where $a_1$, $a_2$ are arbitrary parameters, $f$ and $g$ are independent of these parameters, and the set $\{f,g\}$ is linearly independent. (Here we are ignoring the trivially additive parameter in the potential.) Functional independence for the potential is the additional requirement that the set $\{f,g\}$ is functionally independent. Functional dependence means essentially that the system could be recast as 1-parameter.

What is important to notice here is the occurrence of the first order
condition $V_x+V_y+V_z=0$ for the potential as a consequence of the
Bertrand--Darboux equations. Thus the potential is a~function
of only two functionally independent variables, impossible for nondegenerate
potentials.

We observe the FLD relation
\[ (x+y+z)^2{ {\mathcal S}}_{1}^0-\big(x^2+y^2+z^2\big){{\mathcal S}}_{2}^0+2{{\mathcal S}}_{3}^0-2(x+y+z){{\mathcal S}}_{4}^0-2{{\mathcal S}}_{5}^0=0
\]
obeyed by the purely quadratic terms in the symmetries, i.e., we
have set ${\mathcal S}_{i}={ {\mathcal S}}_{i}^0+W_i$. We show below, in Theorem~\ref{functionallinearindependence1}, that the existence of such an FLD relation implies the existence of a first order condition in the Bertrand--Darboux equations \eqref{BD1}. Furthermore, the 5 quadratic symmetries are functionally dependent:
\[ {\mathcal H}({\mathcal S}_3-{\mathcal S}_5)-\frac{{\mathcal S}_4^2}{2}-\frac{{\mathcal S}_2{\mathcal S}_3}{2}+\frac{(\alpha+\beta+\gamma)}{2}{\mathcal H}-\frac{(\alpha+\beta+\gamma)}{2}{\mathcal S}_2=0.\]
Hence the system defined by \eqref{Cal1} is minimally superintegrable with~4 functionally independent symmetries. We show below, in Corollary~\ref{corollary1}, that this is a generic feature of FLD systems with exactly 5 linearly independent generators.

\subsection{A Minkowski space FLD example}

Here
\begin{equation}\label{Mink_Hamiltonian}
{\mathcal H}=p_1^2+p_2^2+p_3^2 +\alpha(x-z)+\beta(y+{\rm i}z)+\gamma(y+{\rm i}z)^2,
\end{equation}
which admits the 1st order symmetry
\[{\mathcal J}=p_1-{\rm i}p_2+p_3\]
 and the 2nd order symmetries \cite{EM}
\begin{gather*}
{\mathcal S}_{1}= {\mathcal H}=p_1^2+p_2^2+p_3^2 +\alpha(x-z)+\beta(y+{\rm i}z)+\gamma(y+{\rm i}z)^2, \\
 {\mathcal S}_{2}= {\mathcal J}^2,\qquad {\mathcal S}_{3}=p_1^2+\alpha x,\qquad
 {\mathcal S}_{4}= (-{\rm i}p_2+p_3)p_1+(p_3-{\rm i}p_2)^2+\tfrac{\alpha}{2}({\rm i}y-x-z), \\
 {\mathcal S}_{5}= (p_1-{\rm i} p_2+p_3)({\rm i}J_{12}-J_{13})
 -\tfrac{{\rm i}}{2}\alpha yz-\tfrac{{\rm i}}{2}\alpha xy+\tfrac14 \alpha x^2+\tfrac12\alpha xz-\tfrac14 \alpha y^2+\tfrac14 \alpha z^2.
\end{gather*}

The 5 generators are linearly independent and satisfy the FLD relation
\[ ({\rm i}y-z){\mathcal S}_{2}^0+(-{\rm i}y+x+z){ {\mathcal S}}_{4}^0+{{\mathcal S}}_{5}^0=0,\]
where as before $\mathcal{S}_{k}^0$ is the quadratic momentum part of the symmetry $\mathcal{S}_{k}$.

\section{Some theory}\label{4}
Functional linear dependence of a functionally independent maximal set
of symmetries is hard to achieve. We recall the following
result where the system need not be superintegrable \cite{KKM2006}:
\begin{Theorem} \label{functionallinearindependence1}
Let the linearly independent set $\{{\mathcal H}={\mathcal S}_{1},{
{\mathcal S}}_{2},\dots, {\mathcal S}_{t}\}$, $(t>2)$ be a functionally
linearly dependent basis of $2$nd order symmetries for the system
 ${\mathcal H}=\big(p_1^2+p_2^2+p_3^3\big)/\lambda({\bf x})+V={\mathcal H}^0+V$ with nontrivial potential
$V$, i.e., there is a relation
$\sum_{k}f_k({\bf x}) { {\mathcal S}}_{k}^0\equiv 0$ in an open set, where not all
$f_{k}({\bf x})$ are constants, and no such relation holds for the
$f_{k}({\bf x})$ all constant, except if the constants are all zero. $($Here
${\mathcal S}_{i}= {\mathcal S}_{i}^0+W_i$ where the $W_i$ are the potential terms.$)$ Then the
potential must satisfy a first order relation $AV_x+BV_y+CV_z=0$ where
not all of the functions~$A$,~$B$,~$C$ vanish.
\end{Theorem}

\begin{proof} By relabeling, we can express one of the quadratic parts of
the constants of the motion~${ {\mathcal S}}_{0}^0$ as
a linear combination of a linearly independent subset
\begin{equation*}
\big\{{ {\mathcal S}}_{1}^0,\dots,{ {\mathcal S}}_{r}^0, \, 1\le r\le t-1\big\},\end{equation*}
i.e.,
 \begin{equation*} { \mathcal S}_{0}^0=\sum_{\ell=1}^r f_{\ell}(x,y,z) { {\mathcal S}}_{\ell}^0.\end{equation*}

Taking the Poisson bracket of both sides of this equation with $\big(p_1^2+p_2^2+p_3^3\big)/\lambda$ and
using the fact that each of the ${{\mathcal S}}_{h}$ is a constant of the
motion, we obtain the identity
 \begin{equation*}
 \sum_{\ell=1}^r\sum_{i,j=1}^3(\partial_{x_k}f_{\ell})a_{\ell}^{ij}p_ip_jp_k =0,
\end{equation*}
where $(x,y,z)\equiv (x_1,x_2,x_3)$.
It is straightforward to check that this identity can be satisfied if
and only if the functions
\[ c_{k}^{ij}= \sum_{\ell=1}^r (\partial_{x_k}f_{\ell})a_{\ell}^{ij},\qquad 1\le i,j,k\le 3
\]
satisfy the equations
\begin{equation}\label{defeqns1} c_i^{ii}=0, \qquad c_j^{ii}+2c_i^{ij}=0,\quad i\ne j,
\qquad
c_3^{12}+c_1^{23}+c_2^{31}=0.
\end{equation}
Note that $c_k^{ij}=c_k^{ji}$. Corresponding to each of the basis symmetries ${\mathcal
 S}_{\ell}$ there is a linear set $C_{\ell}=0$ of Bertrand--Darboux equations~\cite{KKM2006}. A straightforward substitution into the identity
$C_0-\sum_{\ell=1}^rf_{\ell}({\bf x})C_\ell=0$
yields the relation
\begin{equation}\label{3symmetries}
 \left(\begin{matrix} c^{12}_1- c^{11}_2\\ c^{31}_1-c^{11}_3
\\ c^{31}_2-c^{21}_3\end{matrix}\right)V_1+ \left(\begin{matrix}c^{22}_1-c^{21}_2\\ c^{32}_1-c^{12}_3
\\ c^{32}_2-c^{22}_3\end{matrix}\right)V_2+ \left(\begin{matrix}c^{32}_1-c^{31}_2\\ c^{33}_1-c^{13}_3
\\ c^{33}_2-c^{23}_3\end{matrix}\right)V_3=0.
\end{equation}

These first order differential equations for the potential cannot all vanish identically.
Indeed if they did all vanish then we would have the conditions
\begin{gather*} c^{12}_1=c^{11}_2,\qquad c^{31}_1=c^{11}_3,\qquad
c^{31}_2=c^{21}_3,\qquad c^{22}_1=c^{21}_2,\qquad c^{32}_1=c^{12}_3,\\
c^{32}_2=c^{22}_3,\qquad c^{32}_1=c^{31}_2,\qquad
c^{33}_1=c^{13}_3,\qquad c^{33}_2=c^{23}_3.
\end{gather*}
These conditions, together with conditions (\ref{defeqns1}) show that
$c_i^{jk}=0$ for all $i$, $j$, $k$. Thus we have
$\sum_{\ell=1}^r (\partial_{x_k}f_{\ell})a_{\ell}^{ij}
=0$, $ 1\le i,j,k\le 3$.
Since the set $\big\{{ {\mathcal S}}_{1}^0,\dots,{ {\mathcal S}}_{r}^0\big\}$, is functionally
linearly independent, we have $\partial_{x_k}f_{\ell}\equiv 0$
for $1\le k\le 3$, $1\le \ell\le r$. Hence the $f_{\ell}$ are
constants, which means that
$ { {\mathcal S}}_{0}^0-\sum_{\ell=1}^rf_{\ell} { { \mathcal S}}_{\ell}^0=0$.
 Thus the set $\big\{{\mathcal S}_{0}^0,\dots, {\mathcal S}_{4}^0\big\}$ is linearly dependent. This is a~contradiction!
\end{proof}

This shows that the potential function for any system, superintegrable or not, with a basis of
symmetries that is functionally linearly dependent must satisfy at
least one
nontrivial
first order partial differential equation $AV_x+BV_y+CV_z=0$ where
 the functions~$A$,~$B$,~$C$ are parameter-free. This means that all such
 potentials depend on either one or two functionally independent coordinates.

\begin{Corollary}\label{corollary1} Suppose the system has exactly $5$ linearly independent generators $\{{\mathcal S}_1={\mathcal H},\dots,\allowbreak {\mathcal S}_{5}\}$ and is a functionally
linearly dependent basis of $2$nd order symmetries for the Hamiltonian
 ${\mathcal H}=\big(p_1^2+p_2^2+p_3^3\big)/\lambda({\bf x})+V={\mathcal H}^0+V$ with $3$-parameter potential. Then this set of $5$ generators must be functionally dependent.
\end{Corollary}

\begin{proof} Suppose the set is functionally independent. Then from \cite[equation~(2)]{KKM2007a} at any fixed point there is a potential for any prescribed
values of $V$, $V_x$, $V_y$, $V_z$. However, since the system is FLD the potentials must satisfy $A(x,y,z)V_x+B(x,y,z)V_y+C(x,y,z)V_z=0 $ for $A$, $B$, $C$ not all zero, so the possible derivatives of $V$ are not independent. Contradiction!
\end{proof}

Thus for systems with exactly 5 linearly independent symmetries at most 4 of the 5 FLD generators can form a functionally independent set. However we shall show that there are FLD systems with 2-parameter potentials that admit $>5$ linearly independent and 5 functionally independent 2nd order symmetries in which case Corollary~\ref{corollary1} does not apply.

\begin{Lemma} \label{lemma1} Equations \eqref{defeqns1} imply
\[ \partial_{x_i}\big(c^{ij}_i-c^{ii}_j\big)=0,\qquad \partial_{x_i}\big(c^{ij}_k-c^{ik}_j\big)=0.\]
\end{Lemma}

A new result is
\begin{Theorem}\label{functionallinearindependence2} Under the hypotheses of Theorem~{\rm \ref{functionallinearindependence1}} there exists a $1$st order Killing vector $\mathcal J$ for $\mathcal H$, i.e., $\{{\mathcal J},{\mathcal H\}=\{ \mathcal J},V\}=0$, of the form \[{\mathcal J}=a_1 p_1+a_2p_2+a_3 p_3+a_4(xp_2-yp_1)+a_5( yp_3-zp_2)+a_6( zp_1-xp_z)\]
for some constants $a_j$, not all zero.
\end{Theorem}

\begin{proof} Let
\[ {\mathcal J}=\big(c_1^{12}-c_2^{11}\big)p_1+\big(c_1^{22}-c_2^{12}\big)p_2+\big(c_1^{23}-c_2^{13}\big)p_3={\mathcal J}^x p_1+{\mathcal J}^y p_2+{\mathcal J}^z p_3,\]
so that the first of equations (\ref{3symmetries}) is $\{{\mathcal J},V\}=0$.
From equations (\ref{defeqns1}) and Lemma \ref{lemma1} we can verify that
\begin{gather*} \big\{ {\mathcal J},{\mathcal H}^0\big\}= -\left[\frac{\big(c^{12}_1-c_2^{11}\big)\lambda_1}{\lambda}+\frac{\big(c^{22}_1-c_2^{12}\big)\lambda_2}{\lambda} +\frac{\big(c^{23}_1-c_2^{13}\big)\lambda_3}{\lambda}\right]{\mathcal H}^0\\
\hphantom{\big\{ {\mathcal J},{\mathcal H}^0\big\}}{}=-\frac{1}{\lambda}\big[ 3c_1^{12}\lambda_1-3c_2^{12}\lambda_2+\big(c_1^{23}-c_2^{13}\big)\lambda_3\big] {\mathcal H}^0,
\end{gather*}
so either ${\mathcal J}=0$ or ${\mathcal J}$ is a conformal symmetry of ${\mathcal H}^0$.
However, from Lemma~\ref{lemma1} we see that \begin{equation}\label{confrequirement} \partial_x {\mathcal J}^x=\partial_y {\mathcal J}^y=\partial_z {\mathcal J}^z=0. \end{equation}
The first order conformal symmetries of ${\mathcal H}^0$ are the same as for the case $\lambda=1$,
and the only such symmetries that satisfy the requirements (\ref{confrequirement}) are
linear combinations of $p_1$, $p_2$, $p_3$ and
\[ J_{12}=xp_2-yp_1,\qquad J_{23}= yp_3-zp_2,\qquad J_{13}=xp_3-zp_1,\] and these would be actual symmetries of
${\mathcal H}^0$ (true conformal symmetries such as $x p_1+y p_2+zp_3$ fail this test).
Thus either ${\mathcal J}$ vanishes or it is a 1st order symmetry of~$\mathcal H$.

Analogous constructions and conclusions can be obtained for the 2nd and 3rd of equations~(\ref{3symmetries}). However, at least one of these equations is nonzero.
\end{proof}

Since any Euclidean coordinate transformation applied to the Hamiltonian $\mathcal H$ takes it into one of similar form
\[ {\tilde{\mathcal H}}=\frac{{\tilde p}^2_1+{\tilde p}^2_2+{\tilde p}^2_3}{\tilde \lambda}+{\tilde V},\]
without loss of generality, we can assume that, up to conjugacy, $\mathcal J$ takes one of the five canonical forms:
\begin{gather} p_1,\quad p_1+{\rm i}p_2,\quad xp_2-yp_1,\quad (xp_2-yp_1)+{\rm i}(yp_3-zp_2),\nonumber\\
 (xp_2-yp_1)+{\rm i}(yp_3-zp_2)+p_3+{\rm i}p_1.\label{Jcan}
\end{gather}

With the same assumptions for FLD systems as in Theorem~\ref{functionallinearindependence1}, let ${\mathcal O}_{{\mathcal H}^0} (r)$ be the set of all subsets ${\mathcal B}$ of $\{{\mathcal S}_{1}={\mathcal H},{{\mathcal S}}_{2},\dots, {\mathcal S}_{t}\}$ with $r+1$ elements such that, after relabeling, there is an FLD relation
\begin{equation}\label{OH0expansion} \hat{\mathcal S}_{0}^0=\sum_{\ell=1}^{r} \hat{f}_{\ell}(x,y,z) \hat{\mathcal S}_{\ell}^0,\end{equation}
and such that
\begin{enumerate}\itemsep=0pt
\item[1)] ${\mathcal H}^0,{\mathcal J}^2\in \operatorname{span} {\mathcal B}$,
\item[2)] $\operatorname{span} \mathcal B\subseteq \operatorname{span} \operatorname{Ad}_{\mathcal J}\mathcal B$,
\item[3)] $\mathcal H$ admits a 2-parameter potential.
\end{enumerate}

In this paper we find all superintegrable Hamiltonians $\mathcal H$ on constant curvature spaces for which
${\mathcal O}_{{\mathcal H}^0}(4)\ne \varnothing $. Note that if $\mathcal H$ admits exactly 5 linearly independent symmetries, all cases are included in ${\mathcal O}_{{\mathcal H}^0}(4)$. If $\mathcal H$ admits more than 5 linearly independent 2nd order symmetries we have no proof of completeness but we have not as yet found a verifiable counterexample.

\section{Euclidean space}\label{5}
We first study the possible FLD 2nd order superintegrable systems in 3D complex Euclidean space. Complex metrics were commonly used in the 19th century. Of particular interest for superintegrability and separation of variables are the paper~\cite{K1887} and the book~\cite{B1894}. B{\^o}cher was the first president of the American Mathematical Society.
The metrics are defined as usual as are the curvature conditions but all the variables are complex. Thus a space is conformally flat if the metric can be expressed as $\lambda(x,y,z)\big({\rm d}x^2+{\rm d}y^2+{\rm d}z^2\big)$ for complex variables~$x$,~$y$,~$z$. The advantage is that one complex system can describe several real forms of this system by specializing the coordinates. For example the complex metric ${\rm d}x^2+{\rm d}y^2+{\rm d}z^2$ is Euclidean for $x$, $y$, $z$ real and Minkowski space for $z={\rm i}w$ for $x$, $y$, $w$ real. In this paper, potentials $V(x,y,z)$ that are real for $x$, $y$, $z$ real live on Euclidean space and potentials that are real for $x$, $y$, $w$ real live on Minkowski space. Every potential belongs to one of these classes. Similar remarks are true for the complex 3-sphere, with real forms the real 3-sphere, and the 3-hyperboloids of one and two sheets.

By relabeling, we can express one of the quadratic parts of
the constants of the motion
${{\mathcal S}}_{0}^0$ as
a~linear combination of the quadratic parts of the remaining~$r$ generators through \eqref{OH0expansion}.
Without loss of generality we can reduce to the case where the expansion~(\ref{OH0expansion}) is unique. The generators ${{\mathcal S}_{0}^0},{{\mathcal S}_{1}^0},
{{\mathcal S}_{2}^0},
\dots,{{\mathcal S}_{r}^0}$ are polynomials in $x$, $y$, $z$ of order at most 2 and are linearly independent.
Thus we can solve for the expansion coefficients in the form $f_{\ell}(x,y,z)=s_{\ell}(x,y,z)/s_{0}(x,y,z)$, $\ell =1,\dots,4$ where $s_{0},s_{1},\dots, s_{r}$ are polynomials in $x$, $y$, $z$ of order at most~2. It follows that
\begin{equation}\label{funclindep}\sum_{a_1,a_2,a_3}A(a_1,a_2,a_3)x^{a_1}y^{a_2}z^{a_3}\equiv s_{0}{{\mathcal S}_{0}^0}-\sum_{\ell=1}^4 s_{\ell} {{\mathcal S}_{\ell}^0}=0,\end{equation} where each coefficient $A(a_1,a_2,a_3)$ must vanish. In particular, the sum of all terms homogeneous of degree $n$ must vanish for each $n=0,\dots, 4$: \[B(n)\equiv \sum_{a_1+a_2+a_3=n}A(a_1,a_2,a_3)x^{a_1}y^{a_2}z^{a_3}=0.\]
Each of the generators $ {{\mathcal S}_{r}^0}$ is a linear combination of terms $J_{ij}J_{k\ell}$, (order~2), $J_{ij}p_k$, (order~1) and~$p_ip_j$, (order~0).

Since we have assumed that the expansion~(\ref{OH0expansion}) is unique, there must be only a single term~$B(N)$ that is not identically zero and each $ {{\mathcal S}_{\ell}^0}$ is homogeneous of degree $0$, $1$, or $2$. Thus each $s_{\ell}$ must be homogeneous of degree~$b$ and each ${{\mathcal S}_{\ell}^0}$ must be homogeneous of degree $c=0,1,2$ where $b+c=N$. This greatly restricts the possibilities for~(\ref{funclindep}).

\subsection{Classification criteria}

In the subsequent five subsections we obtain all FLD-superintegrable bases ${\mathcal B}$ on 3D complex Euclidean space that belong to the class ${\mathcal O}_{{\mathcal H}^0}(4)$. Each such basis is associated with a Hamiltonian ${\mathcal H}={\mathcal H}^0+V$ with a two-parameter potential $V$. We emphasize that ${\mathcal B}$ does not necessarily contain all the (momentum parts of) symmetries of ${\mathcal H}$. We compute $V$ as the general solution of the Bertrand--Darboux equations associated with ${\mathcal B}$. However, the Hamiltonians ${\mathcal H}={\mathcal H}^0+V$ obtained in this way may admit additional symmetries not obtained from ${\mathcal B}$. Additionally, we remark that a particular solution $V_{{\rm p}}$ of the general solution~$V$ may correspond to a Hamiltonian with more symmetries than the Hamiltonian with $V$. We make no attempt to classify these special cases.

The classification is performed modulo complex Euclidean transformations: by the discussion in Section~\ref{4}, the Hamiltonian $\mathcal{H}$ must admit one of the first order symmetries in \eqref{Jcan}. Starting from each of the symmetries in~\eqref{Jcan}, which we denote by $\mathcal{J}$, we use the following algorithm to identify FLD-superintegrable systems.
\begin{enumerate}\itemsep=0pt
\item We compute the action of $\operatorname{Ad}_{\mathcal{J}}$ on a basis of second order symmetries of $\mathcal{H}^0$. We use this to construct a generalized eigenbasis (with respect to $\operatorname{Ad}_{\mathcal{J}}$) of such possible second order symmetries.
\item We then consider 5-element subsets ${\mathcal B}$ of this basis and verify that ${\mathcal B}\in{\mathcal O}_{{\mathcal H}^0}(4)$.
\item For each possible action of $\operatorname{Ad}_{\mathcal{J}}$ on $\mathcal{B}$, we identify all possibilities where 1)~the elements of~$\mathcal{B}$ are homogeneous in the spatial variables, in accordance with the discussion in the previous subsection, 2)~the elements of $\mathcal{B}$ are FLD, 3)~the elements $\mathcal{H}^0$ and $\mathcal{J}^2$ are contained in $\operatorname{span}\mathcal{B}$.
\item For each basis satisfying the criteria in the previous step, we use the Bertrand--Darboux equations to compute the corresponding potential. We require that the potential be 2-parameter functionally independent. In this case $\mathcal{B}\in \mathcal{O}_{\mathcal{H}^0}(4)$. We verify that $\mathcal{H}$ is superintegrable: it must admit at least~4 functionally independent symmetries. The final list of the potentials defining such systems is given in Table~\ref{FLDlist}.
\end{enumerate}

In the case of $\mathcal{J}=p_1$, the space of quadratic forms in $\{p_1,p_2,p_3,J_{12},J_{13},J_{23}\}$, modulo the relation $\mathbf{p}\cdot(\mathbf{p}\times\mathbf{x})=0$, provides a generalized eigenbasis of order two symmetries with respect to $\operatorname{Ad}_{\mathcal{J}}$. Hence we provide details of steps 2-4 of the algorithm above and also show that our examples in Section~\ref{section2} are contained in this case.

The computations involved in the cases of the remaining forms in \eqref{Jcan} are lengthier and we provide only the essential details. In all cases we supply the potentials and the algebra generated by the FLD symmetries.

\subsection[First case: ${\mathcal J}=p_1$]{First case: $\boldsymbol{{\mathcal J}=p_1}$}

Here the centralizer of $\mathcal J$ is the group generated by translation in $y$, $z$ and rotation about the $x$-axis. We can use this freedom to simplify the computation. Since $p_1$ is a symmetry the potential must be of the form $V(y,z)$. Any degree two symmetry can be written as a quadratic form in $\{p_1,p_2,p_3,J_{12},J_{13},J_{23}\}$. Due to the triple product identity $\mathbf{p}\cdot(\mathbf{p}\times\mathbf{x})=0$, the space of such quadratic forms has dimension $21-1=20$.

To be concrete, we write a general symmetry as
\begin{equation}\label{generalsymmetry}
\mathcal{S}=RQR^{\rm T}+F_0(x,y,z),
\end{equation}
where
\begin{equation}\label{Rdefinition}
R=(p_1,p_2,p_3,J_{12},J_{13},J_{23})
\end{equation}
 and
\begin{equation*}
Q=\frac12 \left(\begin{matrix}
2a_1 & a_2 & a_3 & a_4 & a_5 & a_6 \\
a_2 & 2 a_7 & a_8 & a_9 & a_{10} & a_{11} \\
a_3 & a_{8} & 2a_{12} & a_{13} & a_{14} & a_{15} \\
a_4 & a_9 & a_{13} & 2a_{16} & a_{17} & a_{18} \\
a_5 & a_{10} & a_{14} & a_{17} & 2a_{19} & a_{20} \\
a_6 & a_{11} & a_{15} & a_{18} & a_{20} & 2a_{21}
\end{matrix}\right).
\end{equation*}
(To get a true basis of second order symmetries of $\mathcal{H}^0$, we set one of $a_6$, $a_{10}$, $a_{13}$ to zero.)

We use the fact that the adjoint action ${\mathcal S}\to \{p_1,{\mathcal S}\}\equiv \operatorname{Ad}_{p_1}{\mathcal S}$
will map the 5-dimensional space of a solution set into itself. Since this action is essentially differentiation with respect to $x$, it is clear that $\operatorname{Ad}_{p_1}^3=0$, so the generalized eigenvalues of $\operatorname{Ad}_{p_1}$ must all be $0$. Thus the possible Jordan canonical forms for the operator $\operatorname{Ad}_{p_1}$ on a generalized eigenbasis of solutions ${\mathcal S}$ are
\begin{gather} \left(\begin{matrix} 0&1&0&0&0\\
0&0&1&0&0\\ 0&0&0&0&0\\ 0&0&0&0&0\\ 0&0&0&0&0 \end{matrix}\right),\qquad
\left(\begin{matrix} 0&1&0&0&0\\
0&0&1&0&0\\ 0&0&0&0&0\\ 0&0&0&0&1\\ 0&0&0&0&0 \end{matrix}\right),\qquad
 \left(\begin{matrix} 0&1&0&0&0\\
0&0&0&0&0\\ 0&0&0&1&0\\ 0&0&0&0&0\\ 0&0&0&0&0 \end{matrix}\right),\nonumber\\
 \left(\begin{matrix} 0&1&0&0&0\\
0&0&0&0&0\\ 0&0&0&0&0\\ 0&0&0&0&0\\ 0&0&0&0&0 \end{matrix}\right),\qquad
\left(\begin{matrix} 0&0&0&0&0\\
0&0&0&0&0\\ 0&0&0&0&0\\ 0&0&0&0&0\\ 0&0&0&0&0 \end{matrix}\right).\label{formabc}
\end{gather}
 We get 5 different forms depending on the smallest integer $k$ such that $\operatorname{Ad}_{p_1}^k=0$. We will consider each of these 5 forms in turn to determine its implications for the generalized eigenbasis of solutions ${\mathcal S}$.

\subsubsection{Form (\ref{formabc}a)}
We first look at the possibilities for form (\ref{formabc}a). In this case $\operatorname{Ad}_{p_1}^2\ne 0$ so that part of the eigenbasis must be $\{ {\mathcal L}, {\mathcal L}_1 ,{\mathcal L}_2\} $, symmetries that generate a chain of length~3.

The action of $\operatorname{Ad}_{p_1}$ is nontrivial on only two of components of $R$ in~\eqref{Rdefinition}:
\begin{equation}\label{nontrivialPoissonbrackets}
\operatorname{Ad}_{p_1} J_{12}=-p_2,\qquad \operatorname{Ad}_{p_1} J_{13}=-p_3.
\end{equation}
The action of $\operatorname{Ad}_{p_1}$ on any monomial in $\mathcal{S}$ can then be determined from \eqref{nontrivialPoissonbrackets} and the Leibniz property. We find that
\begin{gather}\label{triplet1}
\mathcal{L}=a_{16}J_{12}^2+a_{17}J_{12}J_{13}+a_{19}J_{13}^2+W
\end{gather}
(where here and below, the $a_{ij}$ are assumed to be arbitrary parameters) is the most general homogeneous solution of $\operatorname{Ad}_{p_1}^3=0$. Starting from $\mathcal{L}$, a chain is generated with
\begin{gather}
\mathcal{L}_1= \operatorname{Ad}_{p_1}\mathcal{L}=-2a_{16} p_2 J_{12}-a_{17}(p_2 J_{13}+p_3 J_{12})-2a_{19} p_3 J_{13}+W_1,\nonumber \\
\mathcal{L}_2= \operatorname{Ad}_{p_1}\mathcal{L}_1=2a_{16}p_2^2+2a_{17} p_2 p_3+2a_{19}p_3^2+W_2.\label{triplet2}
\end{gather}
 where we omit the expressions for the functions $W$, $W_1$, $W_2$. In addition there must be 2 eigenfunctions of $\operatorname{Ad}_{p_1}$ with eigenvalue $0$ and independent of ${\mathcal L}_2$.

The symmetries that are annihilated by $\operatorname{Ad}_{p_1}$ take the form
\begin{gather}
{\mathcal K}= b_1 p_1^2+b_2 p_1 p_2+b_3 p_1 p_3+b_6 p_1 J_{23}+b_7 p_2^2+b_8 p_2 p_3+b_{11} p_2 J_{23}\nonumber\\
\hphantom{{\mathcal K}=}{} +b_{12} p_3^2+b_{15} p_3 J_{23}+b_{21} J_{23}^2 +U, \label{singlet}
\end{gather}
where the $b_j$, analogous to $a _j$, are constants to be determined, and $U$ is the potential.

${\mathcal L}$ is homogeneous of order 2 in the variables $x$, $y$, $z$. We consider cases for the form of~$\mathcal{L}_2$. A~very special case is that where, by a rotation if necessary, ${\mathcal L}_2$ takes the form where $a_{16}=a_{19}\ne 0$, $a_{17}=0$. Thus we have ${\mathcal L}_2^0=2a_{16} ({\mathcal H}^0-{\mathcal J}^2)$. Always ${\mathcal H}$ can be assumed to be a basis symmetry, so to achieve form~(\ref{formabc}a) we have to select a symmetry $\mathcal K$ that is linearly independent of the 4 forms already exhibited.

If we choose $\mathcal K$ of order 2 in the spatial variables, so ${\mathcal K}=
b_{21}J_{23}^2$ it is straightforward to show that ${\mathcal B}=\mathcal{L},\mathcal{L}_1,\mathcal{L}_2,\mathcal{K},\mathcal{H}$ is an FLD basis. The Bertrand--Darboux equations for $V(y,z)$ and the potentials associated with these symmetries are obtained from requiring
\begin{equation*}
\{\mathcal{H},\mathcal{L}\}=\{\mathcal{H},\mathcal{L}_1\}=\{\mathcal{H},\mathcal{L}_2\}=\{\mathcal{H},\mathcal{K}\}=0.
\end{equation*}

We consider the equations for $V(y,z)$ and $W(x,y,z)$ arising (as coefficients of $p_1$, $p_2$, $p_3$) from $\{\mathcal{H},\mathcal{L}\}$:
\begin{gather*}
a_{16}x z V_z+a_{16}xy V_y+W_x=0, \qquad
 a_{16}x^2V_y-W_y=0,\qquad
 a_{16}x^2V_z-W_z=0.
\end{gather*}
The second and third equations are satisfied when $W(x,y,z)=a_{16}V(y,z)+W_{00}(x)$, where $W_{00}$ is at this point arbitrary. Upon substituting this form for $W$ into the first equation, we observe that we must have $W_{00}(x)=c_1 x^2+ c_2$, for some constants $c_1$, $c_2$, to obtain a well-defined equation for $V(y,z)$. The general solution of the first equation is then
\begin{equation}\label{soln1} V(y,z)=\frac{F(z/y)}{y^2}\end{equation}
for $F$ an arbitrary function (up to an additive constant, $-c_1$, which we set to zero without loss of generality). The Jacobi identity guarantees that this potential is compatible with the symmetries $\mathcal{L}_1$, $\mathcal{L}_2$. We can verify compatibility with $\mathcal{K}$ directly: a function $U$ of $x$, $y$, $z$ can be found so that
$\{\mathcal{H},\mathcal{K}\}=0$.

The Calogero potential \eqref{Cal1} belongs to the class \eqref{soln1}. Indeed, under the Jacobi transformation
\begin{equation} \label{Jacobi}
x=\tfrac{1}{\sqrt{3}}(r_1+r_2+r_3),\qquad y=\tfrac{1}{\sqrt{2}}(r_2-r_1),\qquad z=\tfrac{1}{\sqrt{6}}(2 r_3-r_2-r_1),\end{equation}
we obtain the Calogero potential \eqref{Cal1} in variables $r_1$, $r_2$, $r_3$ by choosing
\[ F(w)=\frac{\beta}{2(1-\sqrt{3}w)^2}+\frac{\gamma}{2(1+\sqrt{3}w)^2}+\frac{\alpha}{2}.\]

If we choose $\mathcal K$ of order 1, so that $
{\mathcal K}=b_{11} p_2 J_{23}+b_{15} p_3 J_{23} +U$
where $|b_{11}|+|b_{15}|>0$, we can verify that the symmetries ${\mathcal B}=\{\mathcal{L},\mathcal{L}_1,\mathcal{L}_2,\mathcal{K},\mathcal{H}\}$ is an FLD basis and solve the Bertrand--Darboux equations to obtain
\begin{equation}\label{soln2} V(y,z)=\frac{F_{\rm p}\big(\frac{y}{z}\big)}{y^2}, \qquad F_{\rm p}(t)\equiv \frac{\beta_1}{(t+q)^2}+\frac{\beta_2 (qt-1)}{ (t+q)^2\sqrt{1+t^2}},
\end{equation}
where $\beta_1$, $\beta_2$ are arbitrary parameters and $q\equiv b_{11}/b_{15}$. Similarly, applying the Jacobi transformation (\ref{Jacobi}) to (\ref{soln2}) we can obtain a solution adapted to translation invariance.

\looseness=-1 If we choose $\mathcal K$ of order 0, there is no 3-parameter solution for the potential.
The other possibilities for $\mathcal L$ of order~2 are that 1)~${\mathcal L}_2$ can be transformed so that $a_{17}=a_{19}=0$ and the one chains are $\mathcal H$ and $p_1^2$, in which case there is no 3-parameter potential, and 2)~${\mathcal L}_2$ can be transformed so that $a_{17}=2ia_{16}$, $a_{19}=-a_{16}$ and the one chains are $\mathcal H$ and~$p_1^2$, which is not FLD.

\subsubsection{Form (\ref{formabc}b)} Here there is one chain of length 3 and one chain of length 2. The general form for the chain of length 3 is (\ref{triplet1})--(\ref{triplet2}) again. The general form for a chain of length 2 is
\begin{gather}
\mathcal{L}_1'= b_9 p_2 J_{12}+b_{10} p_2 J_{13}+b_{13} p_3 J_{12}+b_{14} p_3 J_{13}+b_{18} J_{12}J_{23}+b_{20} J_{13}J_{23}+W_3,\nonumber \\
\mathcal{L}_2'= \operatorname{Ad}_{p_1}\mathcal{L}_1'=-b_9 p_2^2-b_{10}p_2p_3-b_{13}p_2p_3-b_{14}p_3^2-b_{18}p_2J_{23}-b_{20} p_3 J_{23}+W_4,\label{doublet}
\end{gather}
where $W_3$ and $W_4$ are potentials that will play no role in our analysis. We consider three cases based on the order (in the spatial variables) of $\mathcal{L}_1'$

{\bf Case}: ${\mathcal L}'_1$ of order 2. In this case we have $b_9=b_{10}=b_{13}=b_{14}=0$ so that
so that ${\mathcal L}_1'$ takes the form
$ {\mathcal L}_1'=b_{18} J_{12}J_{23}+b_{20} J_{13}J_{23}+W_3$ and~${\mathcal L}_2'$ takes the form $-b_{18}p_2 J_{23}-b_{20}p_3 J_{23}+W_4$.
Since both $\mathcal H$ and $p_1^2$ are of order~0, and since they both must be included in form~(\ref{formabc}b), this case cannot occur.

{\bf Case}: ${\mathcal L}'_1$ of order 1. In this case we have $b_{18}=b_{20}=0$ and
$\mathcal{L}_1'=b_9 p_2 J_{12}+b_{10} p_2 J_{13}+b_{13} p_3 J_{12}+b_{14} p_3 J_{13}+W_3$, $\mathcal{L}_2'=-b_9 p_2^2-b_{10}p_2p_3-b_{13}p_2p_3-b_{14}p_3^2$. However, there is no choice of the surviving parameters $a_j$ and $b_j$ so that $\mathcal{H}$ or $p_1^2$ is contained in $\operatorname{span}\{\mathcal{L}_2,\mathcal{L}_2'\}$ and this case cannot occur.

{\bf Case}: ${\mathcal L}_1'$ of order 0. This case cannot occur since ${\mathcal L}_1'$ vanishes.

Thus we conclude that form (\ref{formabc}b) does not occur.

\subsubsection{Form (\ref{formabc}c)}
Now we have 2 chains of length 2 and one of length~1.
The general form for a chain of length~2 is~\eqref{doublet}. We use the convention that the first chain of length two, $\{\mathcal{L}_1,\mathcal{L}_2\}$, has parameters~$a_j$ and the second chain of length two, $\{\mathcal{L}_1',\mathcal{L}_2'\}$, has parameters~$b_j$.

The general form for a chain of length 1 is \eqref{singlet}.

It is not possible for both ${\mathcal L}_1$ and ${\mathcal L}_1'$ to be of order 2 since then there would only be one symmetry of order 0, not enough to contain both ${\mathcal H}$ and $p_1^2$. We perform case-based analysis on the allowable cases.

{\bf Case}: ${\mathcal L}_1$ of order 2,\ ${\mathcal L}'_1$ of order 1. This implies that ${\mathcal K}$ must be of order 0, so that ${\mathcal H}$ and $p_1^2$ can be contained in the spanning set. We consider the symmetry $\mathcal{L}_2=-a_{18} p_2 J_{23}-a_{20} p_3 J_{23}$.

By rotation of coordinates about the $z$-axis we can achieve one of the forms $a_{20}\ne 0$, $a_{18}=0$ or $a_{20}\ne 0$, $a_{18}=-{\rm i}a_{20}$. For the second form the basis is not FLD, so can be ruled out. For the first form the basis is FLD but fails the requirement of yielding a 2-parameter potential depending on 2 functionally independent coordinates.

{\bf Case}: Both ${\mathcal L}_1$ and ${\mathcal L}'_1$ are of order 1. Then, since $p_1^2$ and $\mathcal H$ are always basis vectors, the remaining basis symmetry $\mathcal{K}$ must be of order 0. It can be chosen as either $p_2^2$ or $(p_2+{\rm i}p_3)^2$. In the 1st case we determine all possible choices of basis vectors such that the set is FLD. There are only 4 general cases and we verify that none of them define a superintegrable system, i.e., yields a 2-parameter potential. In the 2nd case there are 9 possible FLD families, but they all fail the symmetry test.

\subsubsection{Form (\ref{formabc}d)}
Here we have 1 chain of length 2 and 3 chains of length 1.

The general form for a chain of length 2 is \eqref{doublet} while the general form for a chain of length~1 is~\eqref{singlet}.

There are 2 basic cases: 1) ${\mathcal L}_1'$ is of order 2, ${\mathcal L}_2'$ is of order 1 and ${\mathcal K}$ is of orders, 2, 1, or 0; 2)~${\mathcal L}_1'$~is of order 1, ${\mathcal L}_2'$ is of order 0 and ${\mathcal K}$ is of orders, 2, 1, or 0. We check all of the possibilities and find the Hamiltonian $\mathcal{H}=p_x^2+p_y^2+p_z^2+ V(y,z)$, with
\begin{equation}\label{Minksoln}
V(y,z)=b(z-{\rm i}y)+F(z+{\rm i}y),\end{equation}
where $b$ is a free constant and $F$ is an arbitrary function. The corresponding FLD basis is ${\mathcal B}=\big\{{\mathcal H}^0, {\mathcal J}^2,({\rm i}p_2+p_3)^2, p_1({\rm i} p_2+p_3),(z+{\rm i} y)^2 p_1^2, p_1({\rm i}J_{12}+J_{13}) \big\}$. The Minkowski example in Section~\ref{section2} is a special case of potential \eqref{Minksoln}. Indeed, under the complex orthogonal change of coordinates
 \[ x =-2{\rm i} r_1 ,\qquad y=\tfrac12(r_1+r_2-(1-{\rm i})r_3),\qquad z=\tfrac{{\rm i}}{2}(r_1-r_2-(1-{\rm i})r_3) \]
the potential \eqref{Minksoln} becomes that in~\eqref{Mink_Hamiltonian} when we choose $F(w)=\beta w+\gamma w^2$ and $b=\alpha$.

 A special case of (\ref{Minksoln}) with increased symmetry is \begin{equation}\label{Minksolna}
V(y,z)=b_1(z-{\rm i}y)+b_2(z+{\rm i}y)^2.\end{equation}

 Another case is
\begin{equation} \label{soln4a} V(y,z)= b_1z+\frac{b_2}{y^2}.\end{equation}

A third FLD basis is $\big\{{\mathcal H}^0,p_1^2,p_2^2,p_1p_2,p_1J_{12} \big\}$ with corresponding potential
 \begin{equation} \label{soln4b} V(y,z)= by+F(z),\end{equation}
 where $F$ is an arbitrary function, and $b$ is an arbitrary constant. A special case with increased symmetry is
\begin{equation} \label{soln4c} V(y,z)= b_1y+b_2z.\end{equation}

\begin{Remark}The symmetry algebras of the Hamiltonians corresponding the potentials~\eqref{Minksolna} and~\eqref{soln4c} are omitted below due to their complexity.
\end{Remark}

\subsubsection{Form (\ref{formabc}e)}
Here we have 5 chains of length 1. The possibilities are 1) 1 symmetry of order 2, 2 symmetries of order 1 and 2 symmetries of order 0; 2)~1~symmetry of order 2, 1 symmetry of order 1 and 3 symmetries of order 0; 3)~2 symmetries of order 1 and 3 symmetries of order 0; 4) 1 symmetry of order 1 and 4 symmetries of order~0. In all cases the systems are FLD but they do not admit 2-parameter functionally independent potentials.

\subsubsection{Structure algebras}
For the generalized Calogero system (\ref{soln1}) a basis for the generators is
\begin{gather*}
\mathcal{J}=p_1,\qquad {\mathcal S}_{1}=\mathcal{H}= p_1^2+p_2^2+p_3^2+\frac{F\big(\frac{y}{z}\big)}{y^2},\qquad {\mathcal S}_{2}=p_1^2, \\
{\mathcal S}_{3}=\frac{1}{2} J_{23}^2+\frac{F\big(\frac{y}{z}\big)y^2+F\big(\frac{y}{z}\big)z^2}{2y^2}, \qquad {\mathcal S}_{4}=\frac12\big(J_{12}^2+J_{13}^2\big)+\frac{ x^2F\big(\frac{y}{z}\big)}{2y^2}, \\
 {\mathcal S}_{5}=p_2 J_{12}+p_3 J_{13}+\frac{xF\big(\frac{y}{z}\big)}{y^2}.
\end{gather*}
The nonzero commutators of the generators are
\begin{gather}\label{generalizedCalogeroAlgebra}
 \{\mathcal{J},{\mathcal S}_4\}=-{\mathcal S}_5,\qquad \{\mathcal{J},{\mathcal S}_5\}={\mathcal J}^2-{\mathcal H}, \qquad
 \{{\mathcal S}_4,{\mathcal S}_5\}=-2{\mathcal J}{\mathcal S}_3-2{\mathcal J}{\mathcal S}_4,
 \end{gather}
and the functional relationship is{\samepage
\begin{equation}\label{generalizedCalogeroFLD}
 x^2{\mathcal S}_1^0-\big(x^2+y^2+z^2 \big) {\mathcal S}_2^0 +2{\mathcal S}_{4}^0-2x {\mathcal S}_{5}^0=0.
\end{equation}
Note that both $\mathcal H$ and ${\mathcal S}_3$ lie in the center of this algebra.}

For the system (\ref{soln2}) a basis for the 1st and 2nd order generators is
\begin{gather*}
 {\mathcal J}=p_1,\qquad {\mathcal S}_{1}={\mathcal H},\qquad {\mathcal S}_2={\mathcal J}^2,\qquad {\mathcal S}_3=\frac12 J_{23}^2+\frac{F_{\rm p}\big(\frac{y}{z}\big)\big(y^2+z^2\big)}{2y^2},\\
 {\mathcal S}_4=\frac12\big(J_{12}^2+J_{13}^2\big)+\frac{x^2 F_{\rm p}\big(\frac{y}{z}\big)}{2y^2},\qquad {\mathcal S}_5=p_2 J_{12}+p_3 J_{13}+ \frac{x F_{\rm p}\big(\frac{y}{x}\big)}{y^2}, \\
 {\mathcal S}_6=(qJ_{12}+ J_{13})J_{23}+W_6,\qquad {\mathcal S}_7= (qp_2+p_3) J_{23}+W_7,
\end{gather*}
where we omit the complicated forms of the potentials $W_6$, $W_7$. Since the potential-free parts of the generators satisfy \eqref{generalizedCalogeroFLD}
the set of generators is FLD. The subset $\{{\mathcal J},{\mathcal H},{\mathcal S}_{1},\dots,{\mathcal S}_{5}\}$ generates a closed quadratic algebra with nonzero relations \eqref{generalizedCalogeroAlgebra}. However, if any linear combination of~${\mathcal S}_{6}$,~${\mathcal S}_{7}$ is added to the generators, a new 3rd order symmetry is produced that is not a~polynomial in the generators, so the resulting algebra doesn't close at second order.

\begin{Remark}
The set of symmetries $\{\mathcal{S}_{1},\dots,\mathcal{S}_{7}\}$ contains 5 independent symmetries. However, the set of FLD symmetries $\{{\mathcal L},{\mathcal L}_1,{\mathcal L}_2,{\mathcal H},{\mathcal K}\}$, equivalent to $\{\mathcal{S}_{1},{\mathcal S}_2,{\mathcal S}_4,{\mathcal S}_5,{\mathcal S}_7\}$ via a general linear transformation, contains only 4 independent symmetries (as is the maximum possible by Corollary~\ref{corollary1}). The symmetries $\mathcal{S}_{3}$, $\mathcal{S}_{6}$ are obtained in addition to the FLD symmetries by seeking all 2nd order, linearly independent symmetries of the potential~\eqref{soln2}.
\end{Remark}

For the generalized Minkowski system (\ref{Minksoln}) it is convenient to pass from the original variables $\{x,y,z\}$ to new variables $\{X,Y,Z\}$ where
$ X=x$, $Y=z-{\rm i}y$, $Z=z+{\rm i}y$.
The Hamiltonian then can be written as
${\mathcal H}=p_X^2+4p_Yp_Z+bY+F(Z)$.
 The generating symmetries are
 \begin{gather*}
 {\mathcal J} =p_X, \qquad \mathcal{S}_{1}= {\mathcal H}=p_X^2+4p_Yp_Z+bY+F(Z), \qquad \mathcal{S}_{2}=\mathcal{J}^2, \\
 {\mathcal S}_{3}=Zp_X^2-2Xp_Xp_Y-\tfrac12 b X^2,\qquad {\mathcal S}_{4}=p_Xp_Y+\tfrac12 bX,\qquad {\mathcal S}_{5}=p_Y^2+\tfrac12 bZ,
 \end{gather*}
and the nonzero structure relations are
\[\{{\mathcal J},{\mathcal S}_{3}\}=2{\mathcal S}_{4},\qquad \{{\mathcal J},{\mathcal S}_{4}\}=-\tfrac{b}{2},\qquad \{{\mathcal S}_{3},{\mathcal S}_{4}\}=-2{\mathcal J}{\mathcal S}_{5},\]
with ${\mathcal H}$ in the center of the algebra.
The potential-free parts of the generators satisfy
$-z{\mathcal J}^2+{\mathcal S}_{3}^0+2x{\mathcal S}_{4}^0=0$,
so the system is FLD.

For the system (\ref{soln4a}) the generating symmetries are
\begin{gather*}
{\mathcal J}=p_1,\qquad {\mathcal S}_1={\mathcal H}=p_1^2+p_2^2+p_3^2+b_1 z+\frac{b_2}{y^2},\qquad {\mathcal S}_2={\mathcal J}^2, \\
{\mathcal S}_3=p_1 J_{13}+\frac{b_1 x^2}{4},\qquad {\mathcal S}_4=p_1 p_3+\frac{b_1 x}{2},\qquad {\mathcal S}_5=J_{12}^2+\frac{b_2x^2}{y^2}, \\
{\mathcal S}_6=2 p_2 J_{12}+\frac{2b_2 x}{y^2},\qquad {\mathcal S}_7=2p_2^2+\frac{2b_2}{y^2}, \qquad {\mathcal S}_8=p_2 J_{23}+\frac{b_1 y^2}{4}-\frac{b_2 z}{y^2}.
\end{gather*}
Since the potential-free parts of the generators satisfy
$z {\mathcal S}_{2}^0+{\mathcal S}_{3}^0-x {\mathcal S}_4=0$,
the set of generators is FLD. The subset $\{{\mathcal J},{\mathcal S}_1,{\mathcal S}_{2},{\mathcal S}_{3},{\mathcal S}_{4},{\mathcal S}_{7}\}$ generates a closed quadratic algebra with nonzero relations:
\begin{gather*}
\{{\mathcal J},{\mathcal S}_3\}=-{\mathcal S}_4,\qquad \{{\mathcal J},{\mathcal S}_4\}=-\tfrac{b_1}{2},\qquad \{ {\mathcal S}_2,{\mathcal S}_3\}=-2{\mathcal J}{\mathcal S}_4,\\
\{ {\mathcal S}_2,{\mathcal S}_4\}=-b_1 {\mathcal J},\qquad \{{\mathcal S}_3,{\mathcal S}_4\}={\mathcal J}\left({\mathcal S}_1-\tfrac12{\mathcal S}_7-2{\mathcal S}_2\right).
\end{gather*}
However, if any linear combination of ${\mathcal S}_{5}$, ${\mathcal S}_{6}$, ${\mathcal S}_8$ is added to the generators, a new 3rd order symmetry is produced that is not a polynomial in the generators, so the resulting algebra doesn't close at second order.

For the system (\ref{soln4b}) the generating symmetries are
\begin{gather*}
 {\mathcal J}=p_1,\qquad {\mathcal S}_{1}={\mathcal H}=p_1^2+p_2^2+p_3^2+b y+F(z),\qquad {\mathcal S}_{2}=p_1^2,\\
 {\mathcal S}_{3}= -y p_1^2+x p_1 p_2+\frac{b x^2}{4},\qquad {\mathcal S}_{4}=p_1 p_2+ \frac{b x}{2},\qquad
 {\mathcal S}_{5}=p_2^2+b y,
\end{gather*}
and the nonzero structure relations are
\begin{gather*}
 \{ \mathcal{J},\mathcal{S}_3\}=-{\mathcal S}_4,\qquad \{ \mathcal{J},\mathcal{S}_3\}=-\tfrac{b}{2}, \qquad \{\mathcal{S}_2,\mathcal{S}_3\}=-2\mathcal{J}{\mathcal S}_4,\qquad \{\mathcal{S}_2,\mathcal{S}_4\}=-b\mathcal{J}, \\
 \{\mathcal{S}_3,\mathcal{S}_4\}=-\mathcal{J}(\mathcal{S}_2-\mathcal{S}_4),\qquad \{\mathcal{S}_3,\mathcal{S}_5\}=-2 \mathcal{J}{\mathcal S}_4,\qquad \{\mathcal{S}_4,\mathcal{S}_5\}=-b\mathcal{J}.
\end{gather*}
The potential-free parts of the generators satisfy $y \mathcal{J}^2+x\mathcal{S}_3^0-\mathcal{S}_4^0=0$, so the system is FLD.

\subsection[Second case: ${\mathcal J}=p_1+{\rm i}p_2$]{Second case: $\boldsymbol{{\mathcal J}=p_1+{\rm i}p_2}$}\label{symmetry2}

We introduce appropriate new coordinates $\{\eta,\xi, z\}$ where
$ x=\frac12(\xi+\eta)$, $y=\frac{{\rm i}}{2}(\eta-\xi)$, $z=z$.
In the new coordinates the 1st order symmetries for the potential-free case are:
\begin{gather*} p_1+{\rm i}p_2=2p_\eta={\mathcal J},\qquad p_2={\rm i}(p_\xi-p_\eta),\qquad J_{12}={\rm i}(\xi p_\xi-\eta p_\eta),\\
 J_{13}=\tfrac12(\eta+\xi)p_z-z(p_\eta+p_\xi),\qquad J_{23}=\tfrac{{\rm i}}{2}(\xi-\eta)p_z+{\rm i}z(p_\eta-p_\xi).
\end{gather*}
In this case $\operatorname{Ad}_{p_1+{\rm i}p_2}^3=0$. For convenience we prefer to work with $\tilde{J}=p_\eta=(p_1+{\rm i}p_2)/2$. The canonical forms associated with $\operatorname{Ad}_{\tilde{J}}^3=0$ are again~\eqref{formabc}.

\begin{Remark}A basis of second order symmetries in this case is again given by~\eqref{generalsymmetry}. The formulas for the momentum parts of the symmetry operators appearing below are most naturally expressed in terms of $\{p_1,p_2,p_3,J_{12},J_{13},J_{23}\}$, as before. However, the potentials we obtain are most naturally expressed in terms of the new coordinates $\{\eta,\xi,z\}$. We take this approach below and in Sections~\ref{symmetry3} and~\ref{symmetry4}.
\end{Remark}

$\operatorname{Ad}_{\tilde{J}}$ has nontrivial action on three components of $R$ in~\eqref{Rdefinition}:
\begin{equation*}
\operatorname{Ad}_{\tilde{J}} J_{12}={\rm i}(p_1+{\rm i} p_2)/2={\rm i} p_{\eta},\qquad \operatorname{Ad}_{\tilde{J}} J_{13}=-p_3/2,\qquad \operatorname{Ad}_{\tilde{J}} J_{23}=-{\rm i} p_3/2.
\end{equation*}
From here we can construct a convenient generalized eigenbasis for the 20-dimensional space of symmetries:
\begin{gather*}
L_1= -\tfrac12 J_{12}^2, \qquad
L_2=\tfrac{{\rm i}}{2} J_{12}(J_{13}-{\rm i} J_{23}), \qquad
L_3=2J_{13}^2,
\\
M_1= -\tfrac{{\rm i}}{2}(p_1+{\rm i} p_2)J_{12},\qquad
M_2=-\tfrac14(p_1+{\rm i} p_2)(J_{13}-{\rm i} J_{23})-\tfrac{{\rm i}}{2}p_3 J_{12}, \\
M_3= -2p_3 J_{13},\qquad M_4=J_{13}^2+J_{23}^2, \qquad M_5={\rm i} J_{12}(J_{13}+{\rm i} J_{23}),\\
M_6= {\rm i} p_2(J_{13}-{\rm i} J_{23}), \qquad
M_7=-2{\rm i}(p_1-{\rm i} p_2)J_{12}-p_3(J_{13}-{\rm i} J_{23}),
\\
N_1=\tilde{J}^2=\frac14(p_1+{\rm i} p_2)^2, \!\qquad
N_2= \tfrac12 (p_1+{\rm i} p_2) p_3, \!\qquad
N_3= p_3^2, \!\qquad
N_4= -p_3(J_{13}+{\rm i} J_{23}),\\
N_5= -\tfrac12(p_1+{\rm i} p_2)(J_{13}+{\rm i} J_{23}),\qquad
N_6=-{\rm i} p_2 p_3, \qquad
N_7= \mathcal{H}^0=p_1^2+p_2^2+p_3^2,\\
N_8= (J_{13}+{\rm i} J_{23})^2,\qquad
N_9=-\tfrac12(p_1-{\rm i} p_2)(J_{13}+{\rm i} J_{23}),\qquad
N_{10}=\tfrac14(p_1-{\rm i} p_2)^2,
\end{gather*}
where the 3-chains and 2-chains are $\{L_1,M_1,N_1\}$, $\{L_2,M_2,N_2\}$, $\{L_3,M_3,N_3\}$, $\{M_4,N_4\}$, \linebreak $\{M_5,N_5\}$, $\{M_6,N_6\}$, and $\{M_7,N_7\}$. $N_8$, $N_9$, and $N_{10}$ are 1-chains.

\subsubsection{Form (\ref{formabc}a)}
Here we have a 3-chain and two 1-chains, one of which must be $\mathcal{H}^0$. There are two cases to consider. Either the terminal element of the three chain or the second 1-chain must be $N_1=\tilde{\mathcal{J}}^2$.

In the first case, the 3-chain is $\{L_1+\beta_1 M_4+\beta_2 M_5+\gamma N_8,M_1+\beta_1 N_4+\beta_2 N_5,N_1\}$ (where here and below Greek letters with subscripts are arbitrary parameters analogous to the $a_{ij}$ in the previous subsection; they are fixed by requiring certain combinations of them are FLD) and the 1-chain is one of $\mu_1 N_2+\mu_2 N_3+\mu_3 N_6+\mu_4 N_{10}$, $\mu_1 N_4+\mu_2 N_5+\mu_3 N_9$, or $N_8$ (in which case we can take $\gamma=0$ by a canonical form-preserving change of basis). The first 1-chain possibility is FLD when 1)~$\beta_1=-1/4$, $\mu_2=\mu_3=\mu_4=0$, or 2)~$\beta_1=0$, $\gamma=\beta_2^2/2$, $\mu_1=2\beta_2(2\mu_2-2\beta_2^2-1)$, $\mu_3=-2\beta_2$, $\mu_4=1$ or 3) $\beta_1=\mu_3=\mu_4=0$, $\mu_1=4\beta_2\mu_2$, or 4) $\beta_1=\beta_2=\mu_1=\mu_3=\gamma\mu_4=0$. The third subcase with $\gamma=\beta_2^2/2$ and the fourth subcase with $\gamma=\mu_4=0$ lead to the admissible potentials
\begin{equation}\label{case2_potential1}
V(\xi,z)=\frac{b}{\xi^2}+F(q \xi+z)
\end{equation}
and
\begin{equation}\label{case2_potential2}
V(\xi,z)=\frac{b}{\xi^2}+F(z),
\end{equation}
respectively. Note that \eqref{case2_potential2} is special case of~\eqref{case2_potential1} with increased symmetry.

The second 1-chain possibility is FLD when $\mu_1=\mu_3=0$ and $\beta_1=-1/4$ but does not lead to an admissible potential.

The third 1-chain possibility is FLD when $\beta_1=-1/2$ and $\beta_2=0$, leading to the admissible potential
\begin{align}\label{case2_potential3}
V(\xi,z)=\frac{F(z/\xi)}{\xi^2}. \end{align}
In the second case, the 3-chain is $\{\alpha_1 L_1+\alpha_2 L_2+\alpha_3 L_3+\beta_1 M_4+\beta_2 M_5+\gamma N_8,\alpha_1 M_1+\alpha_2 M_2+\alpha_3 M_3+\beta_1 N_4+\beta_2 N_5,\alpha_1 N_1+\alpha_2 N_2+\alpha_3 N_3\}$. This case is not FLD for any choice of parameters.

\subsubsection{Form (\ref{formabc}b)}
Here we have one 3-chain and one 2-chain. The 3-chain must be $\big\{L_1+\beta_1 M_4+\beta_2 M_5+\gamma N_8,M_1+\beta_1 N_4+\beta_2 N_5,N_1=\tilde{\mathcal{J}}^2\big\}$ and the 2-chain must be $\big\{M_7+\mu_1 N_4+\mu_2 N_5+\mu_3 N_9, N_7=\mathcal{H}^0\big\}$. The symmetries are not FLD for any choice of parameters.

\subsubsection{Form (\ref{formabc}c)}
Here we have two 2-chains and a single 1-chain. There are three cases to consider: the terminal elements of the 2-chains are $\tilde{\mathcal{J}}^2$ and $\mathcal{H}^0$, one 2-chain terminates in $\tilde{\mathcal{J}}^2$ and the 1-chain is $\mathcal{H}^0$, one 2-chain terminates with~$\mathcal{H}^0$ and the 1-chain is $\tilde{\mathcal{J}}^2$.

In the first case, the 2-chains are $\{M_1+\beta_1 N_4+\beta_2 N_5+\beta_3 N_9,N_1\}$ and $\{M_7+\gamma_1 N_4+\gamma_2 N_5+\gamma_3 N_9,N_7\}$ and the 1-chain is one of $N_8$, $\mu_1 N_4+\mu_2 N_5+\mu_3 N_9$, $\mu_1 N_2+\mu_3 N_6+\mu_4 N_{10}$. For the first choice of the 1-chain, the symmetries are FLD when $\beta_1=-1/2,$ $\beta_2=\beta_3=0$, but this does not lead to an admissible potential. The second 1-chain possibility is FLD when $\beta_1=-1/4$, $\beta_3=\mu_1=\mu_3=0$, but this does not lead to an admissible potential. For the third 1-chain possibility, the symmetries are FLD when either $\beta_1=\beta_3=\mu_3=\mu_4=0$, $\mu_1=4\beta_2\mu_2$ or $\beta_1=-1/4,$ $\beta_3=\mu_2=\mu_3=\mu_4=0$, but neither corresponds to an admissible potential.

In the second case, one $2$-chain is $\{M_1+\beta_1 N_4+\beta_2 N_5+\beta_3 N_9,N_1\}$ and the second $2$-chain is either $\{\gamma_1 M_1+\gamma_2 M_2+\gamma_3 M_3+\gamma_4 M_6+\gamma_5 M_7+\delta_1 N_4+\delta_2 N_5+\delta_3 N_9, \gamma_1 N_1+\gamma_2 N_2+\gamma_3 N_3+\gamma_4 N_6+\gamma_5 N_7\}$ (we can take $\gamma_1=0$ by a canonical form-preserving change of basis) or $\{\gamma_1 M_4+\gamma_2 M_5+\delta N_8,\gamma_1 N_4+\gamma_2 N_5\}$. To simplify the analysis, we observe that the symmetry $M_1+\beta_1 N_4+\beta_2 N_5+\beta_3 N_9$ leads to an inadmissible potential unless $\beta_3=0$; similarly, if $\gamma_1 N_1+\gamma_2 N_2+\gamma_3 N_3+\gamma_4 N_6+\gamma_5 N_7$ is a symmetry of an admissible potential we must have $\gamma_4=0$. For the first choice of the second 2-chain, we find three sets of FLD symmetries: $\beta_1=\beta_3=\gamma_1=\gamma_4=0$, $\gamma_2=4\beta_2\gamma_3$; $\beta_1=-1/4$, $\beta_3=\gamma_1=\gamma_3=\gamma_4=0$; and $\beta_3=\gamma_1=\gamma_3=\gamma_4=\gamma_5=0$, $\gamma_2=2\delta_3$, but none of these lead to admissible potentials. The second choice of the second 2-chain leads to an FLD basis when $\beta_1=-1/4$, $\beta_3=\gamma_1=0$, but this does not lead to an admissible potential.

In the third case, one $2$-chain is $\{M_7+\beta_1 N_4+\beta_2 N_5+\beta_3 N_9,N_7\}$ and the second $2$-chain is either $\{\gamma_1 M_1+\gamma_2 M_2+\gamma_3 M_3+\gamma_4 M_6+\gamma_5 M_7+\delta_1 N_4+\delta_2 N_5+\delta_3 N_9, \gamma_1 N_1+\gamma_2 N_2+\gamma_3 N_3+\gamma_4 N_6+\gamma_5 N_7\}$ (we can take $\gamma_5=0$ by a canonical form-preserving change of basis) or $\{\gamma_1 M_4+\gamma_2 M_5+\delta N_8,\gamma_1 N_4+\gamma_2 N_5\}$. Using the requirement $\gamma_4=0$ from the second case, we find that the first choice for the second 2-chain does not yield an FLD basis for any choice of parameters. The second choice for the second 2-chain also does not lead to an FLD basis for any choice of parameters.

\subsubsection{Form (\ref{formabc}d)}

Here we have a 2-chain and three 1-chains. There are again three cases to consider: $\tilde{\mathcal{J}}^2$ and $\mathcal{H}^0$ are 1-chains, $\tilde{\mathcal{J}}^2$ is the terminal element of a 2-chain and $H$ is a 1-chain, and $\mathcal{H}^0$ is the terminal element of a $2$-chain and $\tilde{\mathcal{J}}^2$ is a 1-chain.

\looseness=-1 In the first case, the 2-chain is either $\{\alpha_1 M_1+\alpha_2 M_2+\alpha_3 M_3+\alpha_4 M_6+\alpha_5 M_7+\beta_1 N_4+\beta_2 N_5+\beta_3 N_9,\alpha_1 N_1+\alpha_2 N_2+\alpha_3 N_3+\alpha_4 N_6+\alpha_5 N_7\}$ or $\{\alpha_1 M_4+\alpha_2 M_5+\beta N_8,\alpha_1 N_4+\alpha_2 N_5\}$ and the final 1-chain is one of $\mu_1 N_2 +\mu_2 N_3+\mu_3 N_6+\mu_4 N_{10}$, $\mu_1 N_4+\mu_2 N_5+\mu_3 N_9$, $N_8$. To simplify the analysis, it is sometimes useful to find conditions under which the nontrivial 1-chains are compatible (both correspond to the same admissible potential) before searching for FLD systems. For the first choice of the 2-chain where the final 1-chain is order-0, we have the conditions $\alpha_2=-2\alpha_3 \mu_3/\mu_4$ and $\mu_2=\big(\mu_3^3-2\mu_1\mu_4^2+2\mu_3\mu_4^2\big)/4\mu_3\mu_4$ when $\mu_4\neq 0$ (we must also assume $\alpha_3\mu_3\neq 0$ to avoid linear dependence), but this does not lead to an FLD system with admissible potential. When $\mu_4=0$, the 1-chains are incompatible. For the first choice of the 2-chain where the final 1-chain is order-1, we have the compatibility conditions $\alpha_3=0$ or $\alpha_2=2\alpha_3\mu_2/\mu_1$ ($\mu_1\neq 0$); the first of these leads to an FLD system ($\alpha_3=\alpha_4=\alpha_5=\mu_1=\mu_3=0$, $\beta_3=3\alpha_2/2$) with admissible potential
\begin{equation}\label{case2_potential4}
V(\xi,z)=\frac{b_1}{\xi^2}+b_2(q\xi+z)
\end{equation}
and an FLD system ($\alpha_3=\alpha_4=\alpha_5=\mu_1=\mu_3=0$, $\beta_3=-5\alpha_2/2$) with admissible potential
\begin{equation}\label{case2_potential5}
V(\xi,z)=\frac{b_1}{\xi^{2/3}}+\frac{b_2(q\xi+z)}{\xi^{1/3}}.
\end{equation}
In the first choice of the 2-chain where the final 1-chain is order-2, the symmetries are FLD when $\alpha_3=\alpha_4=\alpha_5=0$ and $\beta_3=\alpha_2/2$, but this does not lead to an admissible potential. For the second choice of the 2-chain where the final 1-chain is order-0, the symmetries are FLD when $\alpha_1=\mu_2=\mu_3=\mu_4=0$, but this does not lead to an admissible potential. For the second choice of the 2-chain where the final 1-chain is order-1, imposing $\mu_3=0$ we find that the symmetries are not FLD for any choice of parameters. For the second choice of the 2-chain where the final 1-chain is order-2, the symmetries are not FLD for any choice of parameters.

In the second case, the 2-chain is $\{M_1+\beta_1 N_4+\beta_2 N_5+\beta_3 N_9, N_1\}$ and there are five subcases for the two remaining 1-chains: one order-2 and one order-1 1-chain, one order-2 and one order-0 1-chain, two order-1 1-chains, one order-1 and one order-0 1-chain, and two order-0 1-chains.

In the first subcase, the symmetries are FLD when $\beta_1=-1/4$ and $\beta_3=\mu_1=\mu_3=0$, but this does not lead to an admissible potential.

In the second subcase, the symmetries are FLD when either $\beta_1=-1/2$, $\beta_2=\beta_3=0$ or $\beta_3=\mu_3=\mu_4=0$. From here we obtain three admissible potentials. When $\beta_1=-1/2$, $\beta_2=\beta_3=0$, $\mu_3,\mu_4\neq 0$ and $\mu_2=\big(\mu_3^3-2\mu_1\mu_4^2+2\mu_3\mu_4^2\big)/4\mu_3\mu_4$, we have the potential
\begin{equation}\label{case2_potential6}
V(\xi,z)=b_1\xi+\frac{b_2}{(q\xi+z)^2},
\end{equation}
when $\beta_2=(\mu_1+2\beta_1\mu_1)/4\mu_2$, $\beta_3=\mu_3=\mu_4=0$, we have the potential
\begin{equation}\label{case2_potential7}
V(\xi,z)=\frac{b}{(q\xi+z)^2}+F(\xi),
\end{equation}
and when $\beta_3=\mu_2=\mu_3=\mu_4=0$, we have the potential
\begin{equation}\label{case2_potential8}
V(\xi,z)=\frac{b z}{\xi^3}+F(\xi).
\end{equation}

In the third subcase, we recall that $\mu_1 N_4+\mu_2 N_5+\mu_3 N_9$ only leads to an admissible potential when $\mu_3=0$. Then, by a canonical form-preserving change of basis, we see that $N_4$ and $N_5$ must be independent symmetries. The symmetries are FLD when $\beta_3=0$ and lead to an admissible potential
\begin{equation}\label{case2_potential9}
V(\xi,z)=\frac{b z}{\xi^{3/2}}+F(\xi).
\end{equation}

In the fourth subcase, we write $\mu_1N_4+\mu_2 N_5+\mu_3 N_9$ and $\nu_1N_2+\nu_2 N_3+\nu_3 N_6+\nu_4 N_{10}$ for the order-1 and order-2 1-chains, respectively. The symmetries are FLD when $\beta_1=-1/4$, $\beta_3=\mu_1=\mu_3=0$ or $\beta_3=\nu_3=\nu_4=0$. There are two resulting FLD systems with admissible potentials: $\beta_2=(\mu_2+2\beta_1\mu_2)/2\mu_1$, $\nu_1=2\mu_2\nu_2/\mu_1$, $\beta_3=\mu_3=\nu_3=\nu_4=0$, we obtain a potential equivalent to \eqref{case2_potential7}
and $\beta_3=\mu_1=\mu_3=\nu_3=\nu_2=\nu_4=0$ with
\begin{equation}\label{case2_potential10}
V(\xi,z)=b z\xi^{a}+F(\xi).
\end{equation}

In the fifth subcase, we $\mu_1N_2+\mu_2 N_3+\mu_3 N_6+\mu_4 N_{10}$ and $\nu_1N_2+\nu_2 N_3+\nu_3 N_6+\nu_4 N_{10}$ for the two order-0 1-chains. Assume first that $\mu_4$ and $\nu_4$ are not both zero. Without loss of generality we assume $\mu_4\neq 0$, so we can take $\nu_4=0$ by a canonical form-preserving change of basis. It is then required that $\nu_3=0$ if we are to have an admissible potential. The 1-chains are incompatible unless $\mu_3=-\nu_1 \mu_4/2\nu_2$. When additionally $\nu_1=4\beta_2\nu_2$, $\beta_1=\beta_3=\nu_3=\nu_4=0$, we find an FLD system with admissible potential
\begin{equation}\label{case2_potential11}
V(\xi,z)=b\xi+F(q\xi+z).
\end{equation}
If $\mu_4=\nu_4=0$, we must also have $\mu_3=\nu_3=0$ and we can consider $N_2$ and $N_3$ as independent symmetries. The symmetries are FLD when $\beta_3=0$; when additionally $\beta_1=-1/10$, we find the admissible potential
\begin{equation}\label{case2_potential12}
V(\xi,z)=b \xi z+F(\xi),
\end{equation}
and when additionally $\beta_1=0$, we find the admissible potential
\begin{equation}\label{case2_potential13}
V(\xi,z)=b z+F(\xi).
\end{equation}

In the third case, the 2-chain is $\{M_7+\beta_1 N_4+\beta_2 N_5+\beta_3 N_9, N_1\}$ and there are five subcases for the two remaining 1-chains: one order-2 and one order-1 1-chain, one order-2 and one order-0 1-chain, two order-1 1-chains, one order-1 and one order-0 1-chain, and two order-0 1-chains. The first three subcases are not FLD for any choice of parameters. In the fourth subcase, we write $\mu_1N_4+\mu_2 N_5+\mu_3 N_9$ and $\nu_1N_2+\nu_2 N_3+\nu_3 N_6+\nu_4 N_{10}$ for the order-1 and order-2 1-chains, respectively. The symmetries are FLD when $\mu_1=\mu_3=\nu_2=\nu_3=\nu_4=0$, but this does not lead to an admissible potential. In the fifth subcase, we write $\mu_1 N_2+\mu_2 N_3+\mu_3 N_6+\mu_4 N_{10}$ and $\nu_1 N_2+\nu_2 N_3+\nu_3 N_6+\nu_4 N_{10}$ for the two 1-chains. Compatibility of these 1-chains requires $\mu_3=\mu_4=\nu_3=\nu_4=0$ and we may take $N_2$ and $N_3$ as independent symmetries. However, the simultaneous admissible potential of $N_2$ and $N_3$ is incompatible with the $M_7+\beta_1 N_4+\beta_2 N_5+\beta_3 N_9$ for all choices of parameters.

\subsubsection{Form (\ref{formabc}e)}

Here we have five 1-chains, two of which must be $\mathcal{H}^0$ and $\tilde{\mathcal{J}}^2$. There are seven cases for the three additional 1-chains:
\begin{enumerate}\itemsep=0pt
\item[1)] one order-2 1-chain and two order-1 1-chains,
\item[2)] one order-2, one order-1, and one order-0 1-chain,
\item[3)] one order-2 1-chain and two order-0 1-chains,
\item[4)] three order-1 1-chains,
\item[5)] two order-1 1-chains and one order-0 1-chain,
\item[6)] one order-1 and two order-0 1-chains,
\item[7)] three order-0 1-chains.
\end{enumerate}

In the first case, we write $N_8$, $\mu_1 N_4+\mu_2 N_5+\mu_3 N_9$ and $\nu_1 N_4+\nu_2 N_5+\nu_3 N_9$ for the three 1-chains. The potential is admissible only if $\mu_3=\nu_3=0$, so we may take~$N_4$ and~$N_5$ as independent symmetries. The symmetries are incompatible (do not have a simultaneous admissible potential).

In the second case, we write $N_8$, $\mu_1 N_4+\mu_2 N_5+\mu_3 N_9$ and $\nu_1 N_2+\nu_2 N_3+\nu_3 N_6+\nu_4 N_{10}$ for the three 1-chains. The symmetries are FLD when $\mu_3=\nu_3=\nu_4=0$. When also $\mu_2=\nu_1=0$, we find the potential
\begin{equation}\label{case2_potential14}
V(\xi,z)=\frac{b}{z^2}+F(\xi);
\end{equation}
when also $\nu_2=\mu_1\nu_1/2\mu_2$, we find the admissible potential{\samepage
\begin{equation}\label{case2_potential15}
V(\xi,z)=\frac{b_1\xi^2+b_2z(\mu_1 z+\mu_2\xi)}{\xi^2(2\mu_1 z+\mu_2\xi)^2}+F(\xi),
\end{equation}
which contains \eqref{case2_potential14} as a special case.}

In the third case, we write $N_8$, $\mu_1 N_2+\mu_2 N_3+\mu_3 N_6+\mu_4 N_{10}$ and $\nu_1 N_2+\nu_2 N_3+\nu_3 N_6+\nu_4 N_{10}$ for the 1-chains. We first assume that one of $\mu_4$, $\nu_4$ is nonzero. Without loss of generality we take $\mu_4\neq 0$ so that we may take $\nu_4=0$ by a canonical form-preserving change of basis. We can only have an admissible potential if also $\nu_3=0$. The symmetries are not FLD for any choice of the remaining parameters. We then consider the case where $\mu_3=\mu_4=\nu_3=\nu_4=0$. We can then take $N_2$ and $N_3$ as independent symmetries, but the symmetries are not FLD.

In the fourth case, we can make a canonical-form preserving change of basis and consider~$N_4$,~$N_5$ and $N_9$ as independent symmetries. These symmetries are incompatible (in particular, $N_9$~does not produce an admissible symmetry).

The fifth case is similar to the first case: we may take $N_4$ and $N_5$ as independent symmetries. We write $\mu_1 N_2+\mu_2 N_3+\mu_3 N_6+\mu_4 N_{10}$ for the remaining nontrivial 1-chain. The symmetries are FLD when $\mu_3=\mu_4=0$; when also $\mu_2=0$, we find the admissible potential~\eqref{case2_potential9}.

In the sixth case, we write $\mu_1 N_4+\mu_2 N_5+\mu_3 N_9$, $\nu_1 N_2+\nu_2 N_3+\nu_3 N_6+\nu_4 N_{10}$, and $\sigma_1 N_2+\sigma_2 N_3+\sigma_3 N_6+\sigma_4 N_{10}$ for the three 1-chains. This case is similar to the third case: the two subcases reduce to $\nu_4\neq 0$, $\sigma_3=\sigma_4=0$ and $\nu_2=\nu_3=\nu_4=\sigma_1=\sigma_3=\sigma_4=0$. The first subcase is FLD when also $\mu_1=\sigma_2=0$, but we do not get an admissible potential. The second subcase is FLD and when also $\mu_1=0$, we find the admissible potential \eqref{case2_potential14}.

In the seventh case, we write $\mu_1 N_2+\mu_2 N_3+\mu_3 N_6+\mu_4 N_{10}$, $\nu_1 N_2+\nu_2 N_3+\nu_3 N_6+\nu_4 N_{10}$, and $\sigma_1 N_2+\sigma_2 N_3+\sigma_3 N_6+\sigma_4 N_{10}$ for the three 1-chains. We assume that at least one of $\mu_4$, $\nu_4$, $\sigma_4$ is nonzero. Without loss of generality, we take $\mu_4\neq 0$ so we can make a canonical form-preserving change of basis and take $\nu_4=\sigma_4=0$. The second and third symmetries will only have an admissible potential if also $\nu_3=\sigma_3=0$, so we may also take $\nu_2=\sigma_1=0$: $N_2$ and $N_3$ are independent symmetries. The symmetries are incompatible unless $\mu_4=0$, a contradiction. We next assume $\mu_4=\nu_4=\sigma_4=0$. Then we may consider $N_2$, $N_3$, and $N_6$ as independent symmetries. These symmetries are incompatible.

\subsubsection{Structure algebras}

For the potential \eqref{case2_potential1}, we have the symmetries
\begin{gather*}
 \tilde{\mathcal{J}}=(p_1+{\rm i}p_2)/2 , \qquad \mathcal{S}_{1}=\mathcal{H}=p_1^2+p_2^2+p_3^2+\frac{b}{\xi^2}+F(q \xi+z),\qquad \mathcal{S}_{2}=\tilde{\mathcal{J}}^2, \\
 \mathcal{S}_{3}=L_1+q M_5+\frac{q^{2}}{2}N_8+\frac{b(2qz-\eta)}{2\xi},\qquad \mathcal{S}_{4}=M_1+qN_5+\frac{b}{2\xi}, \\
 \mathcal{S}_{5}=N_3+4q N_2+F(q z+\xi).
\end{gather*}
They satisfy $4(2qz-\eta)\tilde{\mathcal{J}}^2+\xi\mathcal{S}_{1}^0-4\mathcal{S}_{4}^0-\xi\mathcal{S}_{5}^0=0$
and their nonzero commutators are
\begin{gather*}
\big\{\tilde{\mathcal{J}},\mathcal{S}_{3}\big\}=\mathcal{S}_{4},\qquad \big\{\tilde{\mathcal{J}},\mathcal{S}_{4}\big\}=\tilde{\mathcal{J}}^2,\qquad \{\mathcal{S}_{3},\mathcal{S}_{4}\}=-2\tilde{\mathcal{J}}\mathcal{S}_{3}-q^{2}b\tilde{\mathcal{J}}, \\
\{\mathcal{S}_{3},\mathcal{S}_{5}\}=8q^{2}\tilde{\mathcal{J}}\mathcal{S}_{4},\qquad \{\mathcal{S}_{4},\mathcal{S}_{5}\}=8q^{2}\tilde{\mathcal{J}}^3.
\end{gather*}

For the potential~\eqref{case2_potential2}, the symmetries and their FLD relation and algebra are obtained from that of~\eqref{case2_potential1} in the limit $q\to 0$.

For the potential \eqref{case2_potential3}, we have the symmetries
\begin{gather*}
 \tilde{\mathcal{J}}=(p_1+{\rm i}p_2)/2, \qquad \mathcal{S}_{1}=\mathcal{H}=p_1^2+p_2^2+p_3^2+\frac{F(z/\xi)}{\xi^2} ,\qquad \mathcal{S}_{2}=\tilde{\mathcal{J}}^2, \\
 \mathcal{S}_{3}=L_1-\frac12 M_4-\frac{\big(\xi\eta+z^2\big)F(z/\xi)}{2\xi^2},\qquad \mathcal{S}_{4}=M_1-\frac12 N_4+\frac{F(z/\xi)}{2\xi}, \\
 \mathcal{S}_{5}=N_8+F(z/\xi).
\end{gather*}
They satisfy
$4\big(\xi\eta+z^2\big)\tilde{\mathcal{J}}^2-\xi^2 \mathcal{S}_{1}^0+4\xi\mathcal{S}_{4}^0-\mathcal{S}_{5}^0=0$,
and their nonzero commutators are
\begin{gather*}
\big\{\tilde{\mathcal{J}},\mathcal{S}_{3}\big\}=\mathcal{S}_{4}, \qquad \big\{\tilde{\mathcal{J}},\mathcal{S}_{4}\big\}=\tilde{\mathcal{J}}^2 ,\qquad \{\mathcal{S}_{3},\mathcal{S}_{4}\}=-2\tilde{\mathcal{J}} \mathcal{S}_{3}.
\end{gather*}

For the potential \eqref{case2_potential4}, we have the symmetries
\begin{gather*}
\tilde{\mathcal{J}}=(p_1+{\rm i}p_2)/2 , \qquad \mathcal{S}_{1}=\mathcal{H}=p_1^2+p_2^2+p_3^2+\frac{b_1}{\xi^2}+b_2(q\xi +z),\qquad \mathcal{S}_{2}=\tilde{\mathcal{J}}^2, \\
\mathcal{S}_{3}=L_1+q M_5+\frac{q^2}{2}N_8+\frac{b_1(2qz-\eta)}{2\xi} ,\qquad \mathcal{S}_{4}=M_1+\frac{b_1}{2\xi}+\frac{q b_2 \xi^2}{8}, \\
\mathcal{S}_{5}=M_2-q N_4+\frac32 N_9+\frac{b_1 z}{\xi^2}+\frac{b_2(2qz-\eta)\xi}4,\qquad \mathcal{S}_{6}=N_2+\frac{b_2\xi}4, \\
 \mathcal{S}_{7}=N_3+b_2 z,\qquad \mathcal{S}_{8}=N_5-\frac{b_2 \xi^2}8.
\end{gather*}
They satisfy $4\eta\tilde{\mathcal{J}}^2-\xi \mathcal{S}_{1}^0+4 \mathcal{S}_{4}^0+\xi \mathcal{S}_{7}^0=2z \tilde{\mathcal{J}}^2-\xi \mathcal{S}_{6}^0- \mathcal{S}_{8}^0=0$. The subset $\big\{\tilde{\mathcal{J}},\mathcal{S}_{1},\mathcal{S}_{2},\mathcal{S}_{4},\mathcal{S}_{6},\mathcal{S}_{7},\mathcal{S}_{8}\big\}$ generates a closed quadratic algebra with nonzero relations
\begin{gather*}
\big\{\tilde{\mathcal{J}},\mathcal{S}_{4}\big\}=\tilde{\mathcal{J}}^2,\qquad \{\mathcal{S}_{4},\mathcal{S}_{6}\}=-\tilde{\mathcal{J}}\mathcal{S}_{6},\qquad \{\mathcal{S}_{4},\mathcal{S}_{8}\}=-2\tilde{\mathcal{J}}\mathcal{S}_{8}, \\
 \{\mathcal{S}_{6},\mathcal{S}_{7}\}=-b_2\tilde{\mathcal{J}},\qquad \{\mathcal{S}_{6},\mathcal{S}_{8}\}=-2\tilde{\mathcal{J}}^3, \qquad \{\mathcal{S}_{7},\mathcal{S}_{8}\}=-4\tilde{\mathcal{J}}\mathcal{S}_{6}.
\end{gather*}
However, if any linear combination of $\mathcal{S}_{3}$, $\mathcal{S}_{5}$ is added to the generators, a new 3rd order symmetry is produced that is not a polynomial in the generators, so the resulting algebra doesn't close at second order.

For the potential \eqref{case2_potential5}, we have the symmetries
\begin{gather*}
 \tilde{\mathcal{J}}=(p_1+{\rm i}p_2)/2 , \qquad \mathcal{S}_{1}=\mathcal{H}=p_1^2+p_2^2+p_3^2+\frac{b_1}{\xi^{2/3}}+\frac{b_2(q\xi+z)}{\xi^{4/3}},\qquad \mathcal{S}_{2}=\tilde{\mathcal{J}}^2, \\
 \mathcal{S}_{3}=M_1+2N_4-\frac{b_1\xi^{1/3}}{2}-\frac{b_2(q\xi+16z)}{8\xi^{4/3}} ,\\
 \mathcal{S}_{4}=M_2-qN_4-\frac52 N_9 -\frac{b_1 z}{\xi^{2/3}}+\frac{b_2\big(2q\xi z+3\xi\eta-4z^2\big)}{4\xi^{4/3}}, \\
 \mathcal{S}_{5}=N_2-\frac{3b_2}{4\xi^{1/3}},\qquad \mathcal{S}_{6}=N_5-\frac{3b_2\xi^{2/3}}8.
\end{gather*}
They satisfy $2z\tilde{\mathcal{J}}^2-\xi \mathcal{S}_{5}^0-\mathcal{S}_{6}^0=0$. The subset $\big\{\tilde{\mathcal{J}},\mathcal{S}_{1},\mathcal{S}_{2},\mathcal{S}_{3},\mathcal{S}_{5},\mathcal{S}_{6}\big\}$ generates a closed quadratic algebra with nonzero relations
\begin{gather*}
 \big\{ \tilde{\mathcal{J}},\mathcal{S}_{3} \big\}=\tilde{\mathcal{J}}^2,\qquad \{\mathcal{S}_{3},\mathcal{S}_{5}\}=3\tilde{\mathcal{J}}\mathcal{S}_{5}, \qquad
 \{ \mathcal{S}_{3},\mathcal{S}_{6}\}=-6\tilde{\mathcal{J}}\mathcal{S}_{6},\qquad \{ \mathcal{S}_{5},\mathcal{S}_{6}\}=-2\tilde{\mathcal{J}}^3.
\end{gather*}
However, if $\mathcal{S}_{4}$ is added to the generators, a new 3rd order symmetry is produced that is not a~polynomial in the generators, so the resulting algebra doesn't close at second order.

For the potential \eqref{case2_potential6}, we have the symmetries
\begin{gather*}
 \tilde{\mathcal{J}}=(p_1+{\rm i}p_2)/2, \qquad \mathcal{S}_{1}=\mathcal{H}=p_1^2+p_2^2+p_3^2+b_1\xi+\frac{b_2}{(q\xi+z)^2},\\
 \mathcal{S}_{2}=\tilde{\mathcal{J}}^2,\qquad \mathcal{S}_{3}=M_1+qN_5+\frac{b_1 \xi^2}{8}, \\
 \mathcal{S}_{4}=M_3+4q M_2+2q(1-2q^2)N_5+\frac{b_1 z(z+2q\xi)}{2}+\frac{b_2\big(2qz+2q^2\xi-(\xi+\eta)\big)}{(q\xi+z)^2}, \\
 \mathcal{S}_{5}=N_3+4q N_2+\frac{b_2}{(q\xi+z)^2},\qquad \mathcal{S}_{6}=N_4+2q N_5-\frac{b_2\xi}{(q\xi+z)^2}, \qquad \mathcal{S}_{7}=N_8+\frac{b_2\xi^2}{(q\xi+z)^2}, \\
 \mathcal{S}_{8}=N_{10}-2q N_6-2q(1+2q^2)N_2-\frac{b_1(2qz-\eta)}{2}-\frac{b_2 q^2}{(q\xi+z)^2}.
\end{gather*}
They satisfy
\begin{gather*} 4(2qz-\eta)\tilde{\mathcal{J}}^2+\xi \mathcal{S}_{1}^0-4\mathcal{S}_{3}^0-\xi \mathcal{S}_{5}^0=4\big(\xi \eta+z^2\big)\tilde{\mathcal{J}}^2-\xi^2 \mathcal{S}_{1}^0+4\xi \mathcal{S}_{3}^0-2\xi \mathcal{S}_{6}^0-\mathcal{S}_{7}^0=0.
\end{gather*}
The subset $\big\{\tilde{\mathcal{J}},\mathcal{S}_{1},\mathcal{S}_{2},\mathcal{S}_{3},\mathcal{S}_{5},\mathcal{S}_{6},\mathcal{S}_{7}\big\}$ generates a closed quadratic algebra with nonzero relations
\begin{gather*}
 \big\{\tilde{\mathcal{J}},\mathcal{S}_{3}\big\}=\tilde{\mathcal{J}}^2,\qquad \{\mathcal{S}_{3},\mathcal{S}_{5}\}=8q^2\tilde{\mathcal{J}}^3,\qquad \{\mathcal{S}_{3},\mathcal{S}_{6}\}=-\tilde{\mathcal{J}}\mathcal{S}_{6},\qquad \{\mathcal{S}_{3},\mathcal{S}_{7}\}=-2\tilde{\mathcal{J}}\mathcal{S}_{7},\\
\{\mathcal{S}_{5},\mathcal{S}_{6}\}=-4\tilde{\mathcal{J}}\mathcal{S}_{5}-16q^2\tilde{\mathcal{J}}^3, \qquad
 \{\mathcal{S}_{5},\mathcal{S}_{7}\}=-8\tilde{\mathcal{J}}\mathcal{S}_{6},\qquad \{\mathcal{S}_{6},\mathcal{S}_{7}\}=-4\tilde{\mathcal{J}}\mathcal{S}_{7}.
\end{gather*}
However, if any linear combination of $\mathcal{S}_{4}$, $\mathcal{S}_{8}$ is added to the generators, a new 3rd order symmetry is produced that is not a polynomial in the generators, so the resulting algebra doesn't close at second order.

For the potential \eqref{case2_potential7}, we have the symmetries
\begin{gather*}
 \tilde{\mathcal{J}}=(p_1+{\rm i}p_2)/2, \qquad \mathcal{S}_{1}=\mathcal{H}=p_1^2+p_2^2+p_3^2+\frac{b}{(q\xi+z)^2}+F(\xi) ,\qquad \mathcal{S}_{2}=\tilde{\mathcal{J}}^2, \\
 \mathcal{S}_{3}=M_1+qN_5+\frac14\int \xi F'(\xi)\,\mathrm{d}\xi,\qquad \mathcal{S}_{4}=N_3+4q N_2+\frac{b}{(q\xi+z)^2}, \\
 \mathcal{S}_{5}=N_4+2qN_5-\frac{b\xi}{(q\xi+z)^2},\qquad \mathcal{S}_{6}=N_8+\frac{b\xi^2}{(q\xi+z)^2}.
\end{gather*}
They satisfy
\begin{equation*}
4(2qz-\eta)\tilde{\mathcal{J}}^2+\xi\mathcal{S}_{1}^0-4\mathcal{S}_3^0-\xi \mathcal{S}_{4}^0=4\big(z^2+\xi\eta\big)\tilde{\mathcal{J}}^2-\xi^2\mathcal{S}_{1}^0+4\xi\mathcal{S}_{3}^0-2\xi \mathcal{S}_{5}^0-\mathcal{S}_{6}^0=0
\end{equation*}
and their nonzero commutators are
\begin{gather*}
\big\{\tilde{\mathcal{J}},\mathcal{S}_{3}\big\}=\tilde{\mathcal{J}}^2,\qquad \{\mathcal{S}_{3},\mathcal{S}_{4}\}=8q^2\tilde{\mathcal{J}}^3,\qquad \{\mathcal{S}_{3},\mathcal{S}_{5}\}=-\tilde{\mathcal{J}}\mathcal{S}_{5},\qquad \{\mathcal{S}_{3},\mathcal{S}_{6}\}=-2\tilde{\mathcal{J}}\mathcal{S}_{6}, \\
 \{\mathcal{S}_{4},\mathcal{S}_{5}\}=-4\tilde{\mathcal{J}}\mathcal{S}_{4}-16q^2\tilde{\mathcal{J}}^3,\qquad \{\mathcal{S}_{4},\mathcal{S}_{6}\}=-8\tilde{\mathcal{J}}\mathcal{S}_{5},\qquad \{\mathcal{S}_{5},\mathcal{S}_{6}\}=-4\tilde{\mathcal{J}}\mathcal{S}_{6}.
\end{gather*}

For the potential \eqref{case2_potential8}, we have the symmetries
\begin{gather*}
\tilde{\mathcal{J}}=(p_1+{\rm i}p_2)/2, \qquad \mathcal{S}_{1}=\mathcal{H}=p_1^2+p_2^2+p_3^2+\frac{bz}{\xi^3}+F(\xi), \qquad \mathcal{S}_{2}=\tilde{\mathcal{J}}^2, \\
\mathcal{S}_{3}=M_1-\frac12 N_4+\frac{bz}{2\xi^2}+\frac14 \int \xi F'(\xi)\,\mathrm{d}\xi,\qquad \mathcal{S}_{4}=N_2-\frac{b}{8\xi^2}, \\
\mathcal{S}_{5}=N_5+\frac{b}{4\xi},\qquad \mathcal{S}_{6}=N_8+\frac{bz}{\xi}.
\end{gather*}
They satisfy $4\big(\xi \eta+z^2\big)\tilde{\mathcal{J}}^2-\xi^2 \mathcal{S}_{1}^0+4\xi \mathcal{S}_{3}^0-\mathcal{S}_{6}^0=2z \tilde{\mathcal{J}}^2-\xi \mathcal{S}_{4}^0-\mathcal{S}_{5}^0=0$ and their nonzero commutators are
\begin{gather*}
 \big\{\tilde{\mathcal{J}},\mathcal{S}_{3}\big\}=\tilde{\mathcal{J}}^2,\qquad \{\mathcal{S}_{3},\mathcal{S}_{4}\}=-2\tilde{\mathcal{J}}\mathcal{S}_{4},\qquad \{\mathcal{S}_{3},\mathcal{S}_{5}\}=-\tilde{\mathcal{J}}\mathcal{S}_{5}, \\
 \{\mathcal{S}_{4},\mathcal{S}_{6}\}=-4\tilde{\mathcal{J}}\mathcal{S}_{5},\qquad \{\mathcal{S}_{5},\mathcal{S}_{6}\}=b\tilde{\mathcal{J}}.
\end{gather*}

The case of the potential \eqref{case2_potential9} is treated as a special case of~\eqref{case2_potential10} (with $a=-3/2$) below.

We consider the potential \eqref{case2_potential10}:
\begin{equation*}
V(\xi,z)=b z\xi^a+F(\xi), \qquad a\neq -2,-3/2, -1;
\end{equation*}
we cover these exclusions as special cases below. Under our assumptions we have the symmetries
\begin{gather*}
 \tilde{\mathcal{J}}=(p_1+{\rm i}p_2)/2, \qquad \mathcal{S}_{1}=\mathcal{H}=p_1^2+p_2^2+p_3^2+b z\xi^a+F(\xi),\qquad \mathcal{S}_{2}=\tilde{\mathcal{J}}^2, \\
 \mathcal{S}_{3}=M_1-\frac{a}{2(2a+3)}N_4+\frac{2ab z\xi^{a+1}}{4(2a+3)}+\frac14\int \xi F'(\xi)\,\mathrm{d}\xi, \\
 \mathcal{S}_{4}=N_2+\frac{b\xi^{a+1}}{4(1+a)},\qquad \mathcal{S}_{5}=N_5-\frac{b\xi^{a+2}}{4(a+2)}.
\end{gather*}
They satisfy
\begin{equation}\label{case2_potential10_FLD}
2z\tilde{\mathcal{J}}^2-\xi\mathcal{S}_{4}^0-\mathcal{S}_{5}^0=0
\end{equation}
and their nonzero commutators are
\begin{gather*}
\big\{\tilde{\mathcal{J}},\mathcal{S}_{3}\big\}=\tilde{\mathcal{J}}^2,\qquad \{\mathcal{S}_{3},\mathcal{S}_{4}\}=-\frac{3(a+1)}{2a+3}\tilde{\mathcal{J}}\mathcal{S}_{4},\\
 \{\mathcal{S}_{3},\mathcal{S}_{5}\}=-\frac{3(a+2)}{2a+3}\tilde{\mathcal{J}}\mathcal{S}_{5},\qquad \{\mathcal{S}_{4},\mathcal{S}_{5}\}=-2\tilde{\mathcal{J}}^3.
\end{gather*}
In the case $a=-2$ we have the symmetries
\begin{gather*}
 \tilde{\mathcal{J}}=(p_1+{\rm i}p_2)/2, \qquad \mathcal{S}_{1}=\mathcal{H}=p_1^2+p_2^2+p_3^2+\frac{b z}{\xi^2}+F(\xi),\qquad \mathcal{S}_{2}=\tilde{\mathcal{J}}^2, \\
 \mathcal{S}_{3}=M_1-N_4+\frac{b z}{\xi}+\frac14\int \xi F'(\xi)\,\mathrm{d}\xi, \quad \mathcal{S}_{4}=N_2-\frac{b}{4\xi},\qquad \mathcal{S}_{5}=N_5-\frac{b \log \xi}{4}.
\end{gather*}
They satisfy \eqref{case2_potential10_FLD} and their nonzero commutators are
\begin{gather*}
\big\{\tilde{\mathcal{J}},\mathcal{S}_{3}\big\}=\tilde{\mathcal{J}}^2,\qquad \{\mathcal{S}_{3},\mathcal{S}_{4}\}=-3\tilde{\mathcal{J}}\mathcal{S}_{4},\qquad \{\mathcal{S}_{3},\mathcal{S}_{5}\}=-\frac{3b}{4}\tilde{\mathcal{J}},\qquad \{\mathcal{S}_{4},\mathcal{S}_{5}\}=-2\tilde{\mathcal{J}}^3.
\end{gather*}
In the case $a=-3/2$ we have the symmetries
\begin{gather*}
 \tilde{\mathcal{J}}=(p_1+{\rm i}p_2)/2, \qquad \mathcal{S}_{1}=\mathcal{H}=p_1^2+p_2^2+p_3^2+\frac{b z}{\xi^{3/2}}+F(\xi),\qquad \mathcal{S}_{2}=\tilde{\mathcal{J}}^2, \\
 \mathcal{S}_{3}=N_4-\frac{b z}{\xi^{1/2}},\qquad \mathcal{S}_{4}=N_2-\frac{b}{2\xi^{1/2}},\qquad \mathcal{S}_{5}=N_5-\frac{b\xi^{1/2}}{2}.
\end{gather*}
They satisfy \eqref{case2_potential10_FLD} and their nonzero commutators are
\begin{gather*}
\{\mathcal{S}_{3},\mathcal{S}_{4}\}=2\tilde{\mathcal{J}}\mathcal{S}_{4},\qquad \{\mathcal{S}_{3},\mathcal{S}_{5}\}=-2\tilde{\mathcal{J}}\mathcal{S}_{5},\qquad \{\mathcal{S}_{4},\mathcal{S}_{5}\}=-2\tilde{\mathcal{J}}^3.
\end{gather*}
In the case $a=-1$ we have the symmetries
\begin{gather*}
 \tilde{\mathcal{J}}=(p_1+{\rm i}p_2)/2, \qquad \mathcal{S}_{1}=\mathcal{H}=p_1^2+p_2^2+p_3^2+\frac{b z}{\xi}+F(\xi),\qquad \mathcal{S}_{2}=\tilde{\mathcal{J}}^2, \\
 \mathcal{S}_{3}=M_1-\frac12 N_4-\frac{b z}{2}+\int \xi F'(\xi)\,\mathrm{d}\xi,\quad \mathcal{S}_{4}=N_2+\frac{b\log\xi}{4},\qquad \mathcal{S}_{5}=N_5-\frac{b\xi}{4}.
\end{gather*}
They satisfy \eqref{case2_potential10_FLD} and their nonzero commutators are
\begin{gather*}\{\tilde{\mathcal{J}},\mathcal{S}_{3}\}=\tilde{\mathcal{J}}^2,\qquad
\{\mathcal{S}_{3},\mathcal{S}_{4}\}=-\tfrac{3b}{4}\tilde{\mathcal{J}},\qquad \{\mathcal{S}_{3},\mathcal{S}_{5}\}=-3\tilde{\mathcal{J}}\mathcal{S}_{5},\qquad \{\mathcal{S}_{4},\mathcal{S}_{5}\}=-2\tilde{\mathcal{J}}^3.
\end{gather*}

For the potential \eqref{case2_potential11}, we have the symmetries
\begin{gather*}
 \tilde{\mathcal{J}}=(p_1+{\rm i}p_2)/2, \qquad \mathcal{S}_{1}=\mathcal{H}=p_1^2+p_2^2+p_3^2+b \xi+F(q\xi +z),\qquad \mathcal{S}_{2}=\tilde{\mathcal{J}}^2, \\
 \mathcal{S}_{3}=M_1+q N_5+\frac{b\xi^2}8,\qquad \mathcal{S}_{4}=4q N_2+N_3+F(q\xi+z), \\
 \mathcal{S}_{5}=2q(1+2q^2)N_2+2q N_6-N_{10}-\frac{b\eta}{2}-qb z+q^2 F(q\xi+z).
\end{gather*}
They satisfy
\begin{equation*}
4(2qz-\eta)\tilde{\mathcal{J}}^2+\xi \mathcal{S}_{1}^0-4\mathcal{S}_{3}^0-\xi \mathcal{S}_{4}^0=0
\end{equation*}
and their nonzero commutators are
\begin{gather*}
 \big\{\tilde{\mathcal{J}},\mathcal{S}_{3}\big\}=\tilde{\mathcal{J}}^2,\qquad \big\{\tilde{\mathcal{J}},\mathcal{S}_{5}\big\}=\tfrac{b}{2},\qquad \{\mathcal{S}_{3},\mathcal{S}_{4}\}=8q^2\tilde{\mathcal{J}}^3,\\
 \{\mathcal{S}_{3},\mathcal{S}_{5}\}=8q^4\tilde{\mathcal{J}}^3+q^2\tilde{\mathcal{J}}\mathcal{S}_{2} -3q^2\tilde{\mathcal{J}}\mathcal{S}_{4}+2\tilde{\mathcal{J}}\mathcal{S}_{5},\qquad \{\mathcal{S}_{4},\mathcal{S}_{5}\}=-4q^2 b\tilde{\mathcal{J}}.
\end{gather*}

The case of the potential \eqref{case2_potential12} is obtained exactly as a special case of \eqref{case2_potential10} (with $a=1$) above.

For the potential \eqref{case2_potential13} (a special case of~\eqref{case2_potential10} with $a=0$, but with an additional symmetry), we have the symmetries
\begin{gather*}
 \tilde{\mathcal{J}}=(p_1+{\rm i}p_2)/2, \qquad \mathcal{S}_{1}=\mathcal{H}=p_1^2+p_2^2+p_3^2+b z+F(\xi),\qquad \mathcal{S}_{2}=\tilde{\mathcal{J}}^2, \\
 \mathcal{S}_{3}=M_1+\frac14 \int \xi F'(\xi)\,\mathrm{d}\xi,\qquad \mathcal{S}_{4}=N_2+\frac{b\xi}{4}, \qquad
 \mathcal{S}_{5}=N_3-bz,\qquad \mathcal{S}_{6}=N_5-\frac{b\xi^2}8.
\end{gather*}
They satisfy $4\eta \tilde{\mathcal{J}}^2-\xi\mathcal{S}_{1}^0+4\mathcal{S}_{3}^0+\xi \mathcal{S}_{5}^0=2z\tilde{\mathcal{J}}^2-\xi \mathcal{S}_{4}^0-\mathcal{S}_{6}^0=0$ and their nonzero commutators are
\begin{gather*}
\big\{\tilde{\mathcal{J}},\mathcal{S}_{3}\big\}=\tilde{\mathcal{J}}^2,\qquad \{\mathcal{S}_{3},\mathcal{S}_{4}\}=-\tilde{\mathcal{J}}\mathcal{S}_{4},\qquad \{\mathcal{S}_{3},\mathcal{S}_{6}\}=-2\tilde{\mathcal{J}}\mathcal{S}_{6},\\
 \{\mathcal{S}_{4},\mathcal{S}_{5}\}=-b\tilde{\mathcal{J}},\qquad \{\mathcal{S}_{4},\mathcal{S}_{6}\}=-4\tilde{\mathcal{J}}^2,\qquad \{\mathcal{S}_{5},\mathcal{S}_{6}\}=-4\tilde{\mathcal{J}}\mathcal{S}_{4}.
\end{gather*}

For the potential \eqref{case2_potential14}, we have the symmetries
\begin{gather*}
 \tilde{\mathcal{J}}=(p_1+{\rm i}p_2)/2, \qquad \mathcal{S}_{1}=\mathcal{H}=p_1^2+p_2^2+p_3^2+\frac{b}{z^2}+F(\xi),\qquad \mathcal{S}_{2}=\tilde{\mathcal{J}}^2, \\
 \mathcal{S}_{3}=M_1+\frac14 \int \xi F'(\xi)\,\mathrm{d}\xi,\qquad \mathcal{S}_{4}=N_2+\frac{b}{z^2}, \qquad
 \mathcal{S}_{5}=N_4-\frac{b\xi}{z^2},\qquad \mathcal{S}_{6}=N_8+\frac{b\xi^2}{z^2}.
\end{gather*}
They satisfy
\[ 4\big(\xi\eta+z^2\big)\tilde{\mathcal{J}}^2-\xi^2\mathcal{S}_{1}^0+4\xi \mathcal{S}_{3}^0-2\xi \mathcal{S}_{5}^0-\mathcal{S}_{6}^0=4\tilde{\mathcal{J}}^2-\xi \mathcal{S}_{1}^0+4\mathcal{S}_{3}^0+\xi\mathcal{S}_{4}^0=0
\]
and their nonzero commutators are
\begin{gather*}
 \big\{\tilde{\mathcal{J}},\mathcal{S}_{3}\big\}=\tilde{\mathcal{J}}^2,\qquad \{\mathcal{S}_{3},\mathcal{S}_{5}\}=-\tilde{\mathcal{J}}\mathcal{S}_{5},\qquad \{\mathcal{S}_{3},\mathcal{S}_{6}\}=-2\tilde{\mathcal{J}}\mathcal{S}_{6},\\
 \{\mathcal{S}_{4},\mathcal{S}_{5}\}=-4\tilde{\mathcal{J}},\qquad \{\mathcal{S}_{4},\mathcal{S}_{6}\}=-8\tilde{\mathcal{J}}\mathcal{S}_{5},\qquad \{\mathcal{S}_{5},\mathcal{S}_{6}\}=-4\tilde{\mathcal{J}}\mathcal{S}_{6}.
\end{gather*}

For the potential \eqref{case2_potential15}, we consider two cases. In the first case, $\mu_2=0$ and~\eqref{case2_potential15} reduces to~\eqref{case2_potential14} after a redefinition of $F(\xi)$. In the second case, we take $\mu_2\neq 0$, so we define $q=\mu_1/\mu_2$ so that~\eqref{case2_potential15} reduces to
\begin{equation*}
V(\xi,z)=\frac{bz(\xi+qz)}{\xi^2(\xi+2qz)^2}+F(\xi)
\end{equation*}
after a redefinition of $F(\xi)$ and introduction of a new free parameter~$b$. For this potential we have the symmetries
\begin{gather*}
 \tilde{\mathcal{J}}=(p_1+{\rm i}p_2)/2, \qquad \mathcal{S}_{1}=\mathcal{H}=p_1^2+p_2^2+p_3^2+\frac{bz(\xi+qz)}{\xi^2(\xi+2qz)^2}+F(\xi),\qquad \mathcal{S}_{2}=\tilde{\mathcal{J}}^2, \\
 \mathcal{S}_{3}=M_1-\frac12 N_4+\frac{b z(\xi+qz)}{2\xi(\xi+2qz)^2}+\frac14 \int \xi F(\xi)\,\mathrm{d}\xi,\qquad \mathcal{S}_{4}=N_2+\frac{q}{2}N_3-\frac{b}{8(\xi+2qz)^2}, \\
 \mathcal{S}_{5}=N_5+qN_4+\frac{b\xi}{4(\xi+2qz)^2},\qquad \mathcal{S}_{6}=N_8-\frac{b\xi^2}{4q(\xi+2qz)^2}.
\end{gather*}
They satisfy
\[ 4(q\eta-z)\tilde{\mathcal{J}}^2-q\xi\mathcal{S}_{1}^0+4q\mathcal{S}_{3}^0+2\xi \mathcal{S}_{4}^0+2\mathcal{S}_{5}^0=4\big(\xi\eta+z^2\big)\tilde{\mathcal{J}}^2-\xi^2\mathcal{S}_{1}^0+4\xi \mathcal{S}_{3}^0-\mathcal{S}_{6}^0=0
\]
and their nonzero commutators are
\begin{gather*}
\big\{\tilde{\mathcal{J}},\mathcal{S}_{3}\big\}=\tilde{\mathcal{J}}^2,\qquad \{\mathcal{S}_{3},\mathcal{S}_{4}\}=-2\tilde{\mathcal{J}}\mathcal{S}_{4},\qquad \{\mathcal{S}_{3},\mathcal{S}_{5}\}=-\tilde{\mathcal{J}}\mathcal{S}_{5}, \\
 \{\mathcal{S}_{4},\mathcal{S}_{5}\}=-4q\tilde{\mathcal{J}}\mathcal{S}_{4}-2\tilde{\mathcal{J}}^3,\qquad \{\mathcal{S}_{4},\mathcal{S}_{6}\}=-4\tilde{\mathcal{J}}\mathcal{S}_{5},\qquad \{\mathcal{S}_{5},\mathcal{S}_{6}\}=-4q\tilde{\mathcal{J}}\mathcal{S}_{6}.
\end{gather*}

\subsection[Third case: ${\mathcal J}=xp_2-yp_1$]{Third case: $\boldsymbol{{\mathcal J}=xp_2-yp_1}$}\label{symmetry3}

Here the centralizer of $\mathcal J$ is the group generated by translation in $z$ and rotations about the $z$-axis. We can use this freedom to simplify the computation. Since ${\mathcal J}$ is a symmetry the potential must be of the form $V\big(x^2+y^2,z\big)$. A basis for symmetries is again given by \eqref{generalsymmetry}, but as in the previous section, we will construct a more convenient basis by consider the action of $\operatorname{Ad}_{J_{12}}$.

In addition we obtain a series of equations for the first derivatives $\partial_xF_0$, $\partial_yF_0$, $\partial_zF_0$, which lead to Bertrand--Darboux equations for $V\big(x^2+y^2,z\big)$. At the end we have to find~5 linearly independent solutions for~$\mathcal S$ and verify that they admit one functionally linearly dependent solution.

The adjoint action ${\mathcal S}\to \{J_{12},{\mathcal S}\}\equiv \operatorname{Ad}_{J_{12}}{\mathcal S}$
will map the 5-dimensional space of a solution set into itself. This action preserves the order of symmetry operators that are homogeneous in Cartesian coordinates. However, it is also convenient to introduce cylindrical coordinates $\{r,\theta,z\}$ where $x=r\cos(\theta)$, $y=r\sin(\theta)$, $z=z$, and
\begin{gather*} p_1=p_r\cos(\theta)-p_{\theta}\sin(\theta)/r, \qquad p_2=p_r\sin(\theta)+p_{\theta}\cos(\theta)/r, \qquad p_3=p_z,\\
 J_{12}=p_{\theta},\qquad J_{13}=(r p_z-z p_r)\cos\theta+\frac{z \sin\theta \,p_\theta}{r},\qquad J_{23}=(r p_z-zp_r)\sin\theta-\frac{z\cos\theta\, p_\theta}{r}.
\end{gather*}

On the components of $R$ in \eqref{Rdefinition}, $\mathcal{J}=J_{12}$ has the following nontrivial actions:
\begin{gather*}
\operatorname{Ad}_{J_{12}} p_1=p_2,\qquad \operatorname{Ad}_{J_{12}} p_2=-p_1,\qquad \operatorname{Ad}_{J_{12}} J_{13}=J_{23},\qquad \operatorname{Ad}_{J_{12}} J_{23}=-J_{13}.
\end{gather*}
We can use these to construct a basis consisting of eigenvectors of $\operatorname{Ad}_{J_{12}}$. We label the eigenvectors to take advantage of their transformation under rotation: eigenvectors with subscripts~$\pm 2$, $\pm 1$, and~$0$ indicate corresponding eigenvalues of~$\pm 2{\rm i}$,~$\pm {\rm i}$, and~0, respectively (the second subscript, when applicable, distinguishes between multiple eigenvectors of the same order with the same eigenvalue).
A complex eigenbasis for the 6-dimensional space of symmetries of order~2 is
\begin{gather*} L_{0,1}= J_{12}^2,\qquad L_{0,2}= J_{13}^2+J_{23}^2\label{31},\qquad L_{+1} = -\tfrac12 J_{12}(J_{13}-{\rm i} J_{23}),\nonumber \\
 L_{-1}= \tfrac12 J_{12}(J_{13}+{\rm i} J_{23}),\qquad
 L_{-2}= -\tfrac{{\rm i}}{8}(J_{13}+{\rm i} J_{23})^2,\qquad
L_{+2}= \tfrac{{\rm i}}{8}(J_{13}-{\rm i} J_{23})^2.
\end{gather*}
A complex eigenbasis for the 8-dimensional space of symmetries of order 1 is
\begin{gather*}
 M_{0,1}=- p_3 J_{12},\qquad
M_{0,2}=-p_1 J_{13}-p_2 J_{23},\qquad
M_{+1,2}=-\tfrac12 p_3(J_{13}+{\rm i} J_{23}), \nonumber\\
 M_{+1,1}= (p_1+{\rm i} p_2)J_{12},\qquad
M_{-1,2} =\tfrac12 p_3 (J_{13}-{\rm i} J_{23}),\qquad
M_{-1,1}=-\tfrac12J_{12}(p_1-{\rm i} p_2), \nonumber\\
 M_{+2}= -\tfrac14(p_1-{\rm i} p_2)(J_{13}-{\rm i} J_{23}),\qquad
M_{-2}= \tfrac14(p_1+{\rm i} p_2)(J_{13}+{\rm i} J_{23}).
\end{gather*}
A complex eigenbasis for the 6-dimensional space of symmetries of order 0 is
\begin{gather*}
 N_{0,1}=p_3^2,\qquad
N_{0,2}=p_1^2+p_2^2+p_3^2,\qquad
N_{-2}=-\tfrac14(p_1+{\rm i} p_2)^2, \nonumber\\
 N_{+2}= \tfrac14(p_1-{\rm i} p_2)^2,\qquad
N_{-1}= -\tfrac12(p_1+{\rm i} p_2)p_3 ,\qquad
N_{+1}=\tfrac12 (p_1-{\rm i} p_2)p_3.
\end{gather*}

Because $\mathcal{J}$ and $\mathcal{H}$ must be basis vectors, it follows that the possible actions of $\operatorname{Ad}_{p_{\theta}}$ on an eigenbasis are described by the canonical forms
\begin{gather} \left(\begin{matrix}
\lambda_1&0&0&0&0\\
0&\lambda_2&0&0&0\\ 0&0&\lambda_3&0&0\\ 0&0&0&0&0\\ 0&0&0&0&0 \end{matrix}\right),\qquad
 \left(\begin{matrix} \lambda_1&0&0&0&0\\
0&\lambda_2&0&0&0\\ 0&0&0&0&0\\ 0&0&0&0&0\\ 0&0&0&0&0 \end{matrix}\right),\nonumber\\
 \left(\begin{matrix} \lambda_1&0&0&0&0\\
0&0&0&0&0\\ 0&0&0&0&0\\ 0&0&0&0&0\\ 0&0&0&0&0 \end{matrix}\right),\qquad
 \left(\begin{matrix} 0&0&0&0&0\\
0&0&0&0&0\\ 0&0&0&0&0\\ 0&0&0&0&0\\ 0&0&0&0&0 \end{matrix}\right),\label{formfg}
\end{gather}
where $\lambda_j=\pm {\rm i}, \pm 2{\rm i}$. Note that there are a large number of cases to consider. The matrices in~\eqref{formfg} are all diagonal and each contains at least 2 zeros on the diagonal because $\mathcal J$ and $\mathcal H$ must always be included as eigenfunctions. The remaining eigenfunctions correspond to~$3$,~$2$,~$1$ or~0 eigenvalues~$\lambda_j$. All possible choices have to be considered from the eigenfunctions listed above.

\subsubsection{Form (\ref{formfg}a)}
Since the eigenvalues for real Euclidean space must occur in complex-conjugate pairs, a system of this form is only possible for Minkowski space. We examine all such cases and find numerous FLD systems, but none are 2-parameter functionally independent superintegrable.

\subsubsection{Form (\ref{formfg}b)}

We find the following FLD bases and potentials (in each case $F$ is an arbitrary function of its argument and $b$ is an arbitrary parameter)
\begin{gather}\label{pot1} {\mathcal B}=\big\{{\mathcal H}^0,{\mathcal J}^2,L_{+1},L_{-1},L_{0,2}\big\} ,\qquad V(r,z)=F\big(r^2+z^2\big)+\frac{b z}{r\big(r^2+z^2\big)},\\
 \label{pot2} {\mathcal B}=\big\{{\mathcal H}^0,{\mathcal J}^2,L_{+2},L_{-2},L_{0,2}\big\},\qquad V(r,z)=F\big(r^2+z^2\big)+\frac{b}{z^2},\\
 \label{pot3} {\mathcal B}=\big\{{\mathcal H}^0,{\mathcal J}^2,N_{+2},N_{-2},N_{0,1}\big\},\qquad V(r,x)=br^2+F(z),\end{gather}
and
\begin{equation}\label{pot4} {\mathcal B}=\big\{{\mathcal H}^0,{\mathcal J}^2,M_{+1,1},M_{-1,1},N_{0,1}\big\},\qquad V(r,z)=\frac{b}{r}+F(z),\end{equation}

In addition, there is the FLD basis and potential
\begin{equation}\label{pot5}{\mathcal B}=\big\{{\mathcal H}^0,{\mathcal J}^2,N_{+2},N_{-2},N_{0,1}\big\} ,\qquad V(r,z)=b_1\big(4 r^2+z^2+2q z\big)+\frac{b_2}{(z+q)^2},\end{equation}
which is 2-parameter superintegrable.

\subsubsection{Form (\ref{formfg}c)}
Since the eigenvalues for real Euclidean space must occur in complex-conjugate pairs, a system of this form is only possible for Minkowski space. We examine all such systems and find that none are FLD.
\subsubsection{Form (\ref{formfg}d)}
Checking over all possibilities for systems with this eigenvalue form, we find that none are FLD.

\subsubsection{Symmetry algebras}
For the potential \eqref{pot1}, we have the symmetries
\begin{gather*}
 \mathcal{J}=J_{12}, \qquad \mathcal{S}_{1}=\mathcal{H}=N_{0,2}+F\big(r^2+z^2\big)+\frac{b z}{r\big(r^2+z^2\big)},\qquad \mathcal{S}_{2}=\mathcal{J}^2, \\
 \mathcal{S}_{3}=L_{0,2}+\frac{b z}{r},\qquad \mathcal{S}_4=L_{-1}-\frac{{\rm i} b {\rm e}^{-{\rm i}\theta}}{4},\qquad \mathcal{S}_{5}=L_{+1}-\frac{{\rm i}b {\rm e}^{{\rm i}\theta}}{4}.
\end{gather*}
They satisfy ${\rm i} z {\rm e}^{{\rm i}\theta}\mathcal{J}^2+r {\rm e}^{2{\rm i}\theta}\mathcal{S}_{4}^0+r\mathcal{S}_{5}^0=0$
and their nonzero commutators are
\begin{gather*}
 \{\mathcal{J},\mathcal{S}_{4}\}={\rm i}\mathcal{S}_{4},\qquad \{\mathcal{J},\mathcal{S}_{5}\}=-{\rm i}\mathcal{S}_{5},\qquad \{\mathcal{S}_{3},\mathcal{S}_{4}\}=-2{\rm i}\mathcal{J}\mathcal{S}_{4}, \\
 \{\mathcal{S}_{3},\mathcal{S}_{5}\}=2{\rm i}\mathcal{J}\mathcal{S}_{5},\qquad \{\mathcal{S}_{4},\mathcal{S}_{5}\}=\tfrac{{\rm i}}{2}\mathcal{J}\mathcal{S}_{3}-\tfrac{{\rm i}}2\mathcal{J}^3.
\end{gather*}

For the potential \eqref{pot2},
we have the symmetries
\begin{gather*}
 \mathcal{J}=J_{12}, \qquad \mathcal{S}_{1}=\mathcal{H}=N_{0,2}+F\big(r^2+z^2\big)+\frac{b }{z^2},\qquad \mathcal{S}_{2}=\mathcal{J}^2, \\
 \mathcal{S}_{3}=L_{0,2}+\frac{b r^2}{z^2},\qquad \mathcal{S}_{4}=L_{+2}-\frac{{\rm i} b r^2 {\rm e}^{2{\rm i}\theta}}{8z^2},\qquad \mathcal{S}_{5}=L_{-2}-\frac{{\rm i} b r^2 {\rm e}^{-2{\rm i}\theta}}{8z^2}.
\end{gather*}
They satisfy $
 2{\rm i} z^2 {\rm e}^{2{\rm i}\theta}\mathcal{J}^2-{\rm i} r^2 {\rm e}^{2{\rm i}\theta}\mathcal{S}_{3}^0-4r^2\mathcal{S}_{4}^0-4r^2{\rm e}^{4{\rm i}\theta}\mathcal{S}_{5}^0=0$
and their nonzero commutators are
\begin{gather*}
 \{\mathcal{J},\mathcal{S}_{4}\}=-2{\rm i}\mathcal{S}_{4},\qquad \{\mathcal{J},\mathcal{S}_{5}\}=2{\rm i}\mathcal{S}_{5},\qquad \{\mathcal{S}_{3},\mathcal{S}_{4}\}=-4{\rm i}\mathcal{J}\mathcal{S}_{4}, \\
 \{\mathcal{S}_{3},\mathcal{S}_{5}\}=-4{\rm i}\mathcal{J}\mathcal{S}_{5},\qquad \{\mathcal{S}_{4},\mathcal{S}_{5}\}=\tfrac{{\rm i}}{8}\mathcal{J}\mathcal{S}_{3}+\tfrac{b}4\mathcal{J}^3.
\end{gather*}

For the potential \eqref{pot3}, we have the symmetries
\begin{gather*}
 \mathcal{J}=J_{12}, \qquad \mathcal{S}_{1}=\mathcal{H}=N_{0,2}+br^2+G(z), \qquad \mathcal{S}_{2}=\mathcal{J}^2, \\
 \mathcal{S}_{3}=N_{+2}-\frac{br^2{\rm e}^{2{\rm i}\theta}}{4} ,\qquad \mathcal{S}_{4}=N_{-2}+\frac{br^2 {\rm e}^{-2{\rm i}\theta}}{4},\qquad \mathcal{S}_{5}=N_{0,1}+G(z).
\end{gather*}
They satisfy $2{\rm e}^{2{\rm i}\theta}\mathcal{J}^2-r^2{\rm e}^{2{\rm i}\theta}\mathcal{S}_{2}^0-2r^2\mathcal{S}_{3}^2r^2{\rm e}^{4{\rm i}\theta}\mathcal{S}_{4}^0+r^2{\rm e}^{2{\rm i}\theta}\mathcal{S}_{5}^0=0$
and their nonzero commutators are
\begin{gather*}
 \{\mathcal{J},\mathcal{S}_{3}\}=-2{\rm i}\mathcal{S}_{3},\qquad \{\mathcal{J},\mathcal{S}_{4}\}=2{\rm i}\mathcal{S}_{4},\qquad \{\mathcal{S}_{3},\mathcal{S}_{4}\}={\rm i} b\mathcal{J}.
\end{gather*}

For the potential \eqref{pot4}, we have the symmetries
\begin{gather*}
 \mathcal{J}=J_{12}, \qquad \mathcal{S}_{1}=\mathcal{H}=N_{0,2}+F(z)+\frac{b }{r}, \qquad \mathcal{S}_{2}=\mathcal{J}^2\\
 \mathcal{S}_{3}=M_{+1,1}+\frac{{\rm i}b {\rm e}^{{\rm i}\theta}}{4} ,\qquad \mathcal{S}_{4}=M_{-1,1}+\frac{{\rm i} b {\rm e}^{-{\rm i}\theta}}{4},\qquad \mathcal{S}_{5}=N_{0,1}+F(z).
\end{gather*}
They satisfy ${\rm i} {\rm e}^{{\rm i}\theta}\mathcal{J}^2- r\mathcal{S}_{3}^0-r {\rm e}^{2{\rm i}\theta}\mathcal{S}_{4}^0=0$
and their nonzero commutators are
\begin{gather*}
 \{\mathcal{J},\mathcal{S}_{3}\}=-{\rm i}\mathcal{S}_{3},\qquad \{\mathcal{J},\mathcal{S}_{4}\}={\rm i}\mathcal{S}_{4},\qquad \{\mathcal{S}_{3},\mathcal{S}_{4}\}=\tfrac{{\rm i}}2 \mathcal{J}(\mathcal{S}_{5}-\mathcal{S}_{2}).
\end{gather*}

For the potential \eqref{pot5}, we have the symmetries
\begin{gather*}
 \mathcal{J}=J_{12}, \qquad \mathcal{S}_{1}=\mathcal{H}=N_{0,2}+b_1\big(4 r^2+z^2+2q z\big)+\frac{b_2}{(z+q)^2},\qquad \mathcal{S}_{2}=\mathcal{J}^2,\\
 \mathcal{S}_{3}=N_{+2}-b_1 r^2 {\rm e}^{2{\rm i}\theta},\qquad \mathcal{S}_{4}=N_{-2}+b_1 r^2 {\rm e}^{-2{\rm i}\theta},\qquad
 \mathcal{S}_{5}=N_{0,1}+b_1 z(z+2q)+\frac{b_2}{(z+q)^2}, \\
 \mathcal{S}_{6}=M_{+1,2}-qN_{+1}+\frac{r{\rm e}^{{\rm i}\theta}(b_1(z+q)^4-b_2)}{2(z+q)^2} ,\\
 \mathcal{S}_{7}=M_{-1,2}-qN_{-1}-\frac{r{\rm e}^{-{\rm i}\theta}(b_1(z+q)^4-b_2)}{2(z+q)^2}.
\end{gather*}
They satisfy $\mathcal{J}^2-r^2\mathcal{S}_{1}^0-2r^2 {\rm e}^{-2{\rm i}\theta}\mathcal{S}_{5}+2r^2 {\rm e}^{2{\rm i}\theta}\mathcal{S}_{6}+r^2\mathcal{S}_{7}=0$. The subset $\{\mathcal{J},\mathcal{S}_{1},\mathcal{S}_{2},\mathcal{S}_{5},\mathcal{S}_{6},\mathcal{S}_{7}\}$ generates a closed quadratic algebra with nonzero relations
\begin{gather*}
\{\mathcal{J},\mathcal{S}_{3}\}=-2{\rm i} \mathcal{S}_{3},\qquad \{\mathcal{J},\mathcal{S}_{4}\}=2{\rm i} \mathcal{S}_{4},\qquad \{\mathcal{S}_{3},\mathcal{S}_{4}\}=4{\rm i}b_1\mathcal{J}.
\end{gather*}
However, if any linear combination of ${\mathcal S}_{3}$, ${\mathcal S}_{4}$ is added to the generators, a new 3rd order symmetry is produced that is not a polynomial in the generators, so the resulting algebra doesn't close at second order.

\subsection[Fourth case: ${\mathcal J}=J_{12}+{\rm i}J_{23}$]{Fourth case: $\boldsymbol{{\mathcal J}=J_{12}+{\rm i}J_{23}}$}\label{symmetry4}

This case is similar to the second case, treated in Section \ref{symmetry2}. We make the change of variables
\begin{align*}
x= -\rho\big[{\rm e}^{-\theta}+{\rm e}^{\theta}\big(1/4-r^2\big)\big], \qquad
y=-\rho r\exp(\theta), \qquad
z={\rm i}\rho\big[{\rm e}^{-\theta}-{\rm e}^{\theta}\big(1/4+r^2\big)\big],
\end{align*}
so that $p_r=2(J_{12}+ {\rm i}J_{23})$ and
\begin{gather*}
 p_1=\frac{{\rm e}^{2\theta}\big(4r^2-1\big)(p_{\theta}-\rho p_{\rho})-4(p_{\theta}-2r p_r+\rho p_{\rho})}{4\rho {\rm e}^{\theta}}, \qquad p_2=\frac{r{\rm e}^{2\theta}(\rho p_{\rho}-p_{\theta}-p_r)}{\rho {\rm e}^{\theta}}, \\
 p_3={\rm i}\frac{{\rm e}^{2\theta}(4r^2+1)(p_{\theta}-\rho p_{\rho})+4(p_{\theta}-2r p_r+\rho p_{\rho})}{4\rho {\rm e}^{\theta}},\\
 J_{12}=-2r p_{\theta}+\big(\tfrac14+r^2+{\rm e}^{-2\theta}\big)p_r,\qquad J_{13}={\rm i}(r p_r-p_{\theta}), \\
 J_{23}={\rm i}\big(r^2+{\rm e}^{-2\theta}\big) p_r-2{\rm i} rp_{\theta}-\tfrac{{\rm i}}{4}p_r.
\end{gather*}
Similarly to Section \ref{symmetry2}, we prefer to work with $\tilde{J}=p_r$. The action of $\operatorname{Ad}_{\tilde{J}}$ on the elements of~$R$ in~\eqref{Rdefinition} is
\begin{gather*}
 \operatorname{Ad}_{\tilde{J}} p_1=2 p_2,\qquad \operatorname{Ad}_{\tilde{J}} p_2=-2(p_1-{\rm i} p_3),\qquad \operatorname{Ad}_{\tilde{J}} p_3=-2{\rm i} p_2, \\
 \operatorname{Ad}_{\tilde{J}} J_{12}=2{\rm i} J_{13},\qquad \operatorname{Ad}_{\tilde{J}} J_{13}=-2{\rm i}(J_{12}+{\rm i} J_{23}),\qquad \operatorname{Ad}_{\tilde{J}} J_{23}=-2 J_{13}.
\end{gather*}
From here we can construct a convenient generalized eigenbasis of symmetries.

A basis for the six-dimensional space of order-two symmetries is
\begin{gather*}
 L_1=\tfrac{1}{24}J_{12}^2,\qquad L_2=\tfrac{{\rm i}}{6} J_{12}J_{23},\qquad L_3=\tfrac13\big(J_{12}^2-J_{13}^2+{\rm i} J_{12}J_{23}\big), \\
 L_4=2{\rm i}(J_{12}+{\rm i} J_{23})J_{13},\qquad L_5=4(J_{12}+{\rm i}J_{23})^2, \qquad L_6=J_{12}^2+J_{13}^2+J_{23}^2.
\end{gather*}
Here, $\{L_1,L_2,L_3,L_4,L_5\}$ form a chain and $\{L_5,L_6\}\subset\ker\operatorname{Ad}_{p_r}$.

A basis for the eight-dimensional space of order-one symmetries is
\begin{gather*}
 M_1=\tfrac1{24}p_1 J_{12},\qquad M_2= \tfrac1{12}(p_2 J_{12}+{\rm i} p_1 J_{13}),\qquad
M_3=\tfrac{{\rm i}}{6}(p_1 J_{23}+2p_2 J_{13}+p_3 J_{12}), \\
 M_4=-{\rm i} p_1 J_{13}+p_2(J_{12}+{\rm i}J_{23})-p_3 J_{13},\qquad M_5=-4(p_1-{\rm i} p_3)(J_{12}+{\rm i} J_{23}), \\
 M_6=\tfrac12(p_2 J_{12}-{\rm i} p_1 J_{13}), \qquad M_7=-2 p_1 J_{12}-{\rm i} p_1 J_{23}+{\rm i} p_3 J_{12}, \\
 M_8=-2({\rm i} p_1+p_3)J_{13}-2 p_2(J_{12}+{\rm i}J_{23}).
\end{gather*}
Here, we have two separate chains: $\{M_1,M_2,M_3,M_4,M_5\}$ and $\{M_6,M_7,M_8\}$: $\{M_5,M_8\}\subset\ker\operatorname{Ad}_{p_r}$.

A basis for the six-dimensional space of order-zero symmetries is
\begin{gather*}
 N_1=\tfrac{1}{24}p_1^2,\qquad N_2=\tfrac16 p_1 p_2,\qquad N_3=-\tfrac13\big(p_1^2-p_2^2+{\rm i} p_1 p_3\big),\\
 N_4=-2(p_1-{\rm i} p_3)p_2,\qquad N_5=4(p_1-{\rm i}p_3)^2,\qquad N_6=\mathcal{H}^0=p_1^2+p_2^2+p_3^2.
\end{gather*}
Here, $\{N_1,N_2,N_3,N_4,N_5\}$ form a chain and $\{N_5,N_6\}\subset\ker\operatorname{Ad}_{\tilde{J}}$.

The possible canonical forms are
\begin{gather}\label{case4_formab}\allowdisplaybreaks
\left(\begin{matrix}
0 & 1 & 0 & 0 & 0 \\
0 & 0 & 1 & 0 & 0 \\
0 & 0 & 0 & 1 & 0 \\
0 & 0 & 0 & 0 & 1 \\
0 & 0 & 0 & 0 & 0
\end{matrix}\right),\qquad
\left(\begin{matrix}
0 & 1 & 0 & 0 & 0 \\
0 & 0 & 1 & 0 & 0 \\
0 & 0 & 0 & 1 & 0 \\
0 & 0 & 0 & 0 & 0 \\
0 & 0 & 0 & 0 & 0
\end{matrix}\right), \\
 \left(\begin{matrix}
0 & 1 & 0 & 0 & 0 \\
0 & 0 & 1 & 0 & 0 \\
0 & 0 & 0 & 0 & 0 \\
0 & 0 & 0 & 0 & 1 \\
0 & 0 & 0 & 0 & 0
\end{matrix}\right),\qquad
\left(\begin{matrix}
0 & 1 & 0 & 0 & 0 \\
0 & 0 & 1 & 0 & 0 \\
0 & 0 & 0 & 0 & 0 \\
0 & 0 & 0 & 0 & 0 \\
0 & 0 & 0 & 0 & 0
\end{matrix}\right),\label{case4_formcd} \\
\left(\begin{matrix}
0 & 1 & 0 & 0 & 0 \\
0 & 0 & 0 & 0 & 0 \\
0 & 0 & 0 & 1 & 0 \\
0 & 0 & 0 & 0 & 0 \\
0 & 0 & 0 & 0 & 0
\end{matrix}\right), \qquad
\left(\begin{matrix}
0 & 1 & 0 & 0 & 0 \\
0 & 0 & 0 & 0 & 0 \\
0 & 0 & 0 & 0 & 0 \\
0 & 0 & 0 & 0 & 0 \\
0 & 0 & 0 & 0 & 0
\end{matrix}\right), \qquad
\left(\begin{matrix}
0 & 0 & 0 & 0 & 0 \\
0 & 0 & 0 & 0 & 0 \\
0 & 0 & 0 & 0 & 0 \\
0 & 0 & 0 & 0 & 0 \\
0 & 0 & 0 & 0 & 0
\end{matrix}\right).\label{case4_formefg}
\end{gather}

\subsubsection{Form (\ref{case4_formab}a)}
Here we have a 5-chain. This form does not occur because it cannot not contain $\mathcal{H}$ (moreover, $\mathcal{H}$ and $p_r^2$ cannot be in the same chain).
\subsubsection{Form (\ref{case4_formab}b)}
This form can (and must) contain both $\tilde{\mathcal{J}}^2$ and $\mathcal{H}^0$. Because $\mathcal{H}^0$ is not in a nontrivial chain, the basis must be
\[
\big\{\mathcal{H}^0,L_2+\beta L_3+\gamma L_4+\delta_1 L_5+\delta_2 L_6,L_3+\beta L_4+\gamma L_5,L_4+\beta L_5,L_5=p_r^2\big\},
\] but we can take $\beta=\gamma=\delta_1=0$ by a canonical form-preserving change of basis.
The chain $\{L_2+\delta_2 L_6,L_3,L_4,L_5\}$ is FLD but does not correspond to an admissible potential.

\subsubsection{Form (\ref{case4_formcd}c)}
Here we have a 3-chain and a 2-chain. This form does not occur because it cannot not contain~$\mathcal{H}$.

\subsubsection{Form (\ref{case4_formcd}d)}
Here we have a 3-chain and two 1-chains. One of the 1-chains is~$\mathcal{H}^0$. First suppose the second one-chain is $p_r^2$.
Using canonical form-preserving changes of basis when necessary, the possible 3-chains are equivalent to one of $\{N_3,N_4,N_5\}$,
\begin{gather*}
\{\alpha_1 M_3+M_6+\beta_1M_4+\gamma_1 M_5, \alpha_1 M_4+M_7+\beta_1 M_5, \alpha_1 M_5+ M_8\},\\
 \{M_3+\alpha_2 M_6+\beta_2 M_7+\gamma_2 M_8, M_4+\alpha_2 M_7+\beta_2 M_8, M_5+\alpha_2 M_8\}.
\end{gather*}
The first case is FLD and provides the admissible potential
\begin{align}\label{case4_potential_1}
V(\rho,\theta)=b\rho^2+F\big(\rho {\rm e}^{\theta}\big).
\end{align}
The second case is not FLD. The third case is FLD when $\alpha_2=0$ and $\beta_2=\pm 1/2$ but these cases do not provide 2-parameter potentials.

If $\tilde{\mathcal{J}}^2$ is not one of the 1-chains, our basis must contain (after a canonical form-preserving change of basis) $\{L_3+\gamma_2 L_6,L_4,L_5\}$. It is left to chose a second 1-chain, for which there are three possibilities: $L_6$ (in which case we can take $\gamma=0$ by a canonical form-preserving change of basis), $\mu M_5+\nu M_8$, and $N_5$. The first possibility gives an FLD basis and has the admissible potential
\begin{equation}\label{case4_potential_2}
V(\rho,\theta)=\frac{b {\rm e}^{-2\theta}}{\rho^2}+F(\rho).
\end{equation}
The second and third possibilities are FLD when $\gamma=1/3$, but neither leads to a 2-parameter potential.

\subsubsection{Form (\ref{case4_formefg}e)}
Here we have two 2-chains and a 1-chain, which must be $\mathcal{H}^0$. One of the 2-chains must be $\{L_4+\mu L_6,L_5\}$. The possibilities for the other 2-chain are (after canonical form-preserving changes of basis) $\{\alpha M_4+M_7+\gamma M_5,\alpha M_5+M_8\}$ or $\{M_4+\beta M_7+\delta M_8,M_5+\beta M_8\}$. Only the latter (together with $L_5$) is FLD when $\alpha=1$, $\beta=\gamma=0$, and $\delta=-1/2$ but does not yield an admissible potential.

\subsubsection{Form (\ref{case4_formefg}f)}
Here we have a 2-chain and three 1-chains, one of which must be $\mathcal{H}^0$. We first assume that the 2-chain is $\{L_4+\mu L_6,L_5\}$. There are then four ways to choose the remaining two one chains: $\{N_5,\alpha M_5+\beta M_8\}$, $\{\alpha M_5+\beta M_8,L_6\}$ (in which case we take $\mu=0$), $\{N_5,L_6\}$ (again, $\mu=0$), or $\{M_5,M_8\}$. The first case is FLD when $\alpha=0$ but does yield an admissible potential. The second, third, and fourth cases are not FLD.

If the 2-chain is not $\{L_4+\mu L_6,L_5\}$, one of the 1-chains must be $L_5=\tilde{\mathcal{J}}^2$. Then we have one 1-chain ($N_5$, $L_6$, or $\mu M_5+\nu M_8$) and one 2-chain ($\{N_4,N_5\}$, $\{\alpha M_4+M_7+\gamma M_5,\alpha M_5+M_8\}$, or $\{M_4+\beta M_7+\delta M_8,M_5+\beta M_8\}$) to choose. There are several FLD bases but only one leads to an admissible potential
\begin{gather}\label{case4_potential_3}
V(\rho,\theta)=\frac{b {\rm e}^{-3\theta}}{\rho}+F\big(\rho {\rm e}^{\theta}\big).
\end{gather}

\subsubsection{Form (\ref{case4_formefg}g)}
This case consists of five 1-chains, two of which must be $\mathcal{H}^0$ and $\tilde{\mathcal{J}}^2$. There are therefore three subcases to consider: the remaining symmetries are either $\{L_6,\alpha M_5+\beta M_8,N_5\}$, $\{L_6,M_5,M_8\}$, or $\{M_5,M_8,N_5\}$. The first and third cases are FLD in certain cases but the corresponding potentials do not have 2 independent parameters.

\subsubsection{Structure algebras}
For the potential \eqref{case4_potential_1}, we have the symmetries
\begin{gather*}
 \tilde{\mathcal{J}}=2(J_{12}+{\rm i} J_{23}), \qquad \mathcal{S}_{1}=\mathcal{H}=N_6+b\rho^2+F\big(\rho {\rm e}^{\theta}\big),\qquad \mathcal{S}_{2}=\tilde{\mathcal{J}}^2, \\
 \mathcal{S}_{3}=L_5+\frac{-\big[4+{\rm e}^{2\theta}\big(1-12r^2\big)\big]b\rho^2-4F\big(\rho {\rm e}^{\theta}\big)}{24}, \qquad \mathcal{S}_{4}=N_4-b\rho^2r{\rm e}^{2\theta},\\ \mathcal{S}_{5}=N_5+b\rho^2{\rm e}^{2\theta}.
\end{gather*}
They satisfy $\rho^2 {\rm e}^{2\theta}\big(4{\mathcal{S}}_{1}^0+24 {\mathcal{S}}_{3}^0+24 {\mathcal{S}}_{4}^0+\big(1+12r^2\big){\mathcal{S}}_{5}^0\big)-12\tilde{\mathcal{J}}^2=0$
and their nonzero commutators are
\begin{gather*}
\big\{\tilde{\mathcal{J}},\mathcal{S}_{3}\big\}=\mathcal{S}_{4},\qquad \big\{\tilde{\mathcal{J}},\mathcal{S}_{4}\big\}=\mathcal{S}_{5},\qquad
\{\mathcal{S}_{3},\mathcal{S}_{4}\}=b\tilde{\mathcal{J}}.
\end{gather*}

For the potential \eqref{case4_potential_2}, we have the symmetries
\begin{gather*}
 \tilde{\mathcal{J}}=2(J_{12}+{\rm i} J_{23}),\qquad \mathcal{S}_{1}=\mathcal{H}=N_6+\frac{b{\rm e}^{-2\theta}}{\rho^2}+F(\rho),\qquad \mathcal{S}_{2}=\tilde{\mathcal{J}}^2+v_0, \\
 \mathcal{S}_{3}=L_3+\frac{b r^2}{2}-\frac{b{\rm e}^{-2\theta}}{3},\qquad \mathcal{S}_{4}= L_4-b r, \qquad \mathcal{S}_{5}=L_6+b {\rm e}^{-2\theta}.
\end{gather*}
They satisfy $\big(1+12{\rm e}^{-2\theta}-12r^2\big) \tilde{\mathcal{J}}^2-12{\mathcal{S}}_{3}^0-12r{\mathcal{S}}_{4}^0-4{\mathcal{S}}_{5}^0=0$
and their nonzero commutators are
\begin{gather*}
 \big\{\tilde{\mathcal{J}},\mathcal{S}_{3} \big\}=\mathcal{S}_{4}, \qquad \big\{\tilde{\mathcal{J}},\mathcal{S}_{4} \big\}=\mathcal{S}_{2}, \qquad
 \{\mathcal{S}_{3},\mathcal{S}_{4} \}=-2\tilde{\mathcal{J}}\mathcal{S}_{3}+\tfrac13\tilde{\mathcal{J}}\mathcal{S}_{5} +\tfrac16\tilde{\mathcal{J}}^3+\tfrac{b}{12}\tilde{\mathcal{J}}.
\end{gather*}

For the potential \eqref{case4_potential_3}, we have the symmetries
\begin{gather*}
 \tilde{\mathcal{J}}=2(J_{12}+{\rm i} J_{23}),\qquad \mathcal{S}_{1}=\mathcal{H}=N_6+\frac{b{\rm e}^{-3\theta}}{\rho}+F\big({\rm e}^{\theta}\rho\big) ,\qquad \mathcal{S}_{2}=\tilde{\mathcal{J}}^2, \\
 \mathcal{S}_{3}=M_4-\frac12M_8-\frac{b\rho r^2}{2\rho}, \qquad \mathcal{S}_{4}=M_5+b r,\qquad \mathcal{S}_{5}=N_5-\frac{2b {\rm e}^{-\theta}}{\rho}.
\end{gather*}
They satisfy $\tilde{\mathcal{J}}^2+\rho {\rm e}^{\theta} \hat{\mathcal{S}}_{3}^0+\rho r {\rm e}^{\theta}\hat{\mathcal{S}}_{4}^0=0$,
and their nonzero commutators are
\begin{gather*}
 \big\{\tilde{\mathcal{J}},\mathcal{S}_{3}\big\}=\mathcal{S}_{4},\qquad \big\{\tilde{\mathcal{J}},\mathcal{S}_{4}\big\}=-b, \qquad
\{\mathcal{S}_{3}, \mathcal{S}_{4}\}=-\tfrac12\tilde{\mathcal{J}}\mathcal{S}_{4}.
\end{gather*}

\subsection[Fifth case: ${\mathcal J}=-{\rm i}J_{12}+J_{23}-{\rm i}p_1+p_3$]{Fifth case: $\boldsymbol{{\mathcal J}=-{\rm i}J_{12}+J_{23}-{\rm i}p_1+p_3}$}\label{symmetry5}

This case does not occur for complex Euclidean systems since the symmetry $\mathcal J$ is not homogeneous.

\subsection{Additional comments} \label{5.5}
We note that for all of the systems classified we can find a complete integral for the Hamilton--Jacobi equation.
For example, the system~(\ref{case2_potential10}) with $a=-3$, has the Hamilton--Jacobi equation
\begin{gather*}
4\frac{\partial S}{\partial\xi}\frac{\partial S}{\partial\eta}+\left(\frac{\partial S}{\partial z}\right)^2+\frac{b z}{\xi^3}+F(\xi)=E.
\end{gather*}
For this equation, we find the complete integral
\begin{gather*}
S(\xi,\eta,z)= \frac{b^2}{768 c_1^3 \xi^3}+\frac{b\big(2c_1^2z+c_2\xi\big)}{16 c_1^3 \xi^2}+c_1 \eta+\frac{c_1^2(4c_2 z+E\xi)-c_2^2\xi}{4c_1^3}
 -\frac1{4c_1}\int F(\xi)\,\mathrm{d}\xi, 
\end{gather*}
where $c_1$, $c_2$ are arbitrary constants and another constant arises from the indefinite integral of~$F$. The corresponding Schr\"{o}dinger equation
\begin{gather*}
\left(4\partial_{\xi}\partial_{\eta}+\partial_z^2+\frac{bz}{\xi^3}+F(\xi)\right)\psi=E\psi
\end{gather*}
has the solution
\begin{equation*}
\psi(\xi,\eta,z)=\exp S(\xi,\eta,z).
\end{equation*}

The symmetry algebras for these FLD superintegrable systems don't always close. However the symmetries always provide some information about the classical trajectories of solutions of the Hamilton--Jacobi equation. If a superintegrable system is functionally independent the trajectories are uniquely determined, However, if the system is FLD then we can solve for one of the constants of the motion in terms of the others. Thus a 2-parameter manifold can be computed from the symmetries such that the trajectories of solutions of the Hamilton--Jacobi equation must lie on this manifold.

\begin{table}[th!]\centering\small \caption{Summary of the FLD systems.}\label{FLDlist}

\vspace{1mm} \begin{tabular}{ cl }
\hline
\\
(\ref{soln1})& $V(y,z)=\frac{F(z/y)}{y^2} $\\
(\ref{soln2})&$ V(y,z)= \frac{\beta_1}{y^2}+\frac{1}{(b_{11}y+b_{15}z)^2}\left(\beta_3+\frac{(-b_{11}z+b_{15}y)\beta_2}{\sqrt{y^2+z^2}}+\frac{(2y^2-z^2)\beta_1b_{15}^2-2 \beta_1b_{11}b_{15}yz}{y^2}\right)$\\
(\ref{Minksoln})&$V(y,z)=b(z-{\rm i}y)+F(z+{\rm i}y)$\\
(\ref{Minksolna})&$V(y,z)=b_1(z-{\rm i}y)+b_2(z+{\rm i}y)^2$\\
(\ref{soln4a})& $V(y,z)= b_1z+\frac{b_2}{y^2}$\\
(\ref{soln4b})& $V(y,z)= cy+F(z)$\\
(\ref{soln4c})& $V(y,z)= c_1y+c_2z$\\
(\ref{case2_potential1}) & $V(\xi,z)=\frac{b}{\xi^2}+F(q \xi+z)$\\
 (\ref{case2_potential2}) & $V(\xi,z)=\frac{b}{\xi^2}+F(z)$ \\
 (\ref{case2_potential3})& $V(\xi,z)=\frac{F(z/\xi)}{\xi^2} $\\
(\ref{case2_potential4}) &$V(\xi,z)=\frac{b}{\xi^2}+b_2(q\xi+z)$ \\
 (\ref{case2_potential5}) &$V(\xi,z)=\frac{b_1}{\xi^{2/3}}+\frac{b_2(q\xi+z)}{\xi^{4/3}}$\\
(\ref{case2_potential6}) & $V(\xi,z)=b_1\xi+\frac{b_2}{(q\xi+z)^2}$\\
(\ref{case2_potential7})& $V(\xi,z)=\frac{b}{(q\xi+z)^2}+F(\xi)$ \\
 (\ref{case2_potential8})&$V(\xi,z)=\frac{b z}{\xi^3}+F(\xi)$ \\
 (\ref{case2_potential9}) & $V(\xi,z)=\frac{b z}{\xi^{3/2}}+F(\xi)$ \\
 (\ref{case2_potential10})& $V(\xi,z)=b z\xi^{a}+F(\xi)$ \\
 (\ref{case2_potential11}) & $V(\xi,z)=b\xi+F(q\xi+z)$ \\
(\ref{case2_potential12})&$V(\xi,z)=b \xi z+F(\xi)$ \\
 (\ref{case2_potential13}) & $V(\xi,z)=b z+F(\xi)$ \\
 (\ref{case2_potential14}) & $V(\xi,z)=\frac{b}{z^2}+F(\xi) $ \\
 (\ref{case2_potential15}) & $V(\xi,z)=\frac{b_1\xi^2+b_2z(\mu_1 z+\mu_2\xi)}{\xi^2(2\mu_1 z+\mu_2\xi)^2}+F(\xi)$ \\
(\ref{pot1})& $ V(r,z)=F\big(r^2+z^2\big)+\frac{b z}{r (r^2+z^2 )} $ \\
(\ref{pot2})&$ V(r,z)=F\big(r^2+z^2\big)+\frac{b}{z^2} $ \\
 (\ref{pot3})& $ V(r,x)=br^2+F(z)$ \\
 (\ref{pot4})& $ V(r,z)=\frac{b}{r}+F(z)$ \\
 (\ref{pot5})& $ V(r,z)=\frac{c_1(4b_5r^2+b_5z^2+2b_3z)}{4a_1b_5}-\frac{c_2}{(2b_5(a_1+2k_1)(b_5z+b_3)^2} $ \\
 (\ref{case4_potential_1})& $V(\rho,\theta)=b\rho^2+F\big(\rho {\rm e}^{\theta}\big)$\\
 (\ref{case4_potential_2})& $V(\rho,\theta)=\frac{b {\rm e}^{-2\theta}}{\rho^2}+F(\rho)$ \\
(\ref{case4_potential_3}) & $V(\rho,\theta)=\frac{b {\rm e}^{-3\theta}}{\rho}+F\big(\rho {\rm e}^{\theta}\big)$\\
\hline
\end{tabular}\end{table}

\section{The complex 3-sphere}\label{6}

 We choose a standardized
Cartesian-like coordinate system $\{x,y,z\}$ on the 3-sphere such that
the Hamiltonian is
\begin{equation*}
{\mathcal H}=\left(1+\frac{r^2}{4}\right)^2\big(p_x^2+p_y^2+p_z^2\big)+V,
\end{equation*}
where $r^2=x^2+y^2+z^2$.
These coordinates can be related to the standard realization of the
sphere via complex coordinates ${\bf s}=(s_1,s_2,s_3,s_4)$ such that
$\sum_{j=1}^4 s_j^2=1$ and ${\rm d}s^2=\sum_j {\rm d}s_j^2$ via
\begin{equation*}
s_1=\frac{4x}{4+r^2},\qquad s_2=\frac{4y}{4+r^2},\qquad
s_3=\frac{4z}{4+r^2},\qquad s_4=\frac{4-r^2}{4+r^2}
\end{equation*}
with inverse
$ x=2s_1/(1+s_4)$, $ y=2s_2/(1+s_4)$, $z=2s_3/(1+s_4)$.
 A basis of Killing vectors for the zero
potential system is $J_h$, $K_h$, $h=1,2,3$ where
\begin{gather*}
J_{23}=yp_z-zp_y,\qquad J_{31}=zp_x-xp_z,\qquad J_{12}=xp_y-yp_x,\nonumber\\
K_1=\left(1+\frac{x^2-y^2-z^2}{4}\right)p_x+\frac{xy}{2}p_y+\frac{xz}{2}p_z,\nonumber\\
K_2=\left(1+\frac{y^2-x^2-z^2}{4}\right)p_y+\frac{xy}{2}p_x+\frac{yz}{2}p_z,\nonumber\\
K_3=\left(1+\frac{z^2-x^2-y^2}{4}\right)p_z+\frac{xz}{2}p_x+\frac{yz}{2}p_y.
\end{gather*}
 The relation between this basis and the
standard basis of rotation generators on the sphere
$ I_{\ell m}=s_\ell p_m-s_m p_\ell=-I_{m\ell}$
is
\begin{equation*}
J_{23}=I_{23},\qquad J_{13}=I_{13},\qquad J_{12}=I_{12},\qquad K_1=I_{41},\qquad K_2=I_{42},\qquad
K_3=I_{43}.
\end{equation*}

To solve the classification problem for Hamiltonians in the class ${\mathcal O}_{{\mathcal H}^0}(4)$ on the complex 3-sphere we can use methods analogous to those for Euclidean space. From result (\ref{Jcan}) applied to the 3-sphere we see that, up to conjugacy, there are just 2 cases to consider: ${\mathcal J}=J_{12}$ and ${\mathcal J}=J_{12}+{\rm i}J_{23}$. The details are complicated but we find that there are no class ${\mathcal O}_{{\mathcal H}^0}(4)$ FLD superintegrable systems on the complex 3-sphere in this class. To save space we do not provide the details here. They can be found in the online paper \cite{ArXiv2}.

\section{Conclusions}\label{7} This paper is part of a program to classify all 2nd order superintegrable classical and quantum systems on 3-dimensional conformally flat complex manifolds. We have worked out the basic structure theory for certain FLD-superintegrable systems on these manifolds and classified all such systems on constant curvature spaces that are in the class ${\mathcal O}_{{\mathcal H}^0}(4)$. There turn out to be no such systems on the complex 3-sphere~\cite{ArXiv2}. The remaining systems to classify are highly degenerate, admitting at least 6 linearly independent symmetries and as yet we have found no verifiable examples. For complex Euclidean space we list the 2-parameter potentials in Table~\ref{FLDlist}. In most of the cases the potential depends on at least one arbitrary function. The key to the classification is a proof that all such systems admit a 1st order symmetry.

\subsection*{Acknowledgements}
We thank the referees for suggestions that improved our presentation, in particular one referee for pointing out a gap in our original classification. B.K.B.\ acknowledges support from the G\"oran Gustafsson Foundation. W.M.\ was partially supported by a grant from the Simons Foundation (\#~412351 to Willard Miller, Jr.).

\pdfbookmark[1]{References}{ref}
\LastPageEnding


\begin{thebibliography}{99}
\footnotesize\itemsep=0pt

\bibitem{BCR}
Benenti S., Chanu C., Rastelli G., The super-separability of the three-body
 inverse-square {C}alogero system, \href{https://doi.org/10.1063/1.533369}{\textit{J.~Math. Phys.}} \textbf{41} (2000),
 4654--4678.

\bibitem{ArXiv2}
Berntson B.K., Kalnins E.G., Miller Jr. W., Nonexistence of {C}alogero-like 2nd
 order superintegrable systems on the complex 3-sphere, available at
 \url{https://www-users.math.umn.edu/~mille003/FLDpaperArk_e.pdf}.

\bibitem{B1894}
B\^ocher M., \"Uber die Riehenentwickelungen der Potentialtheory, B.G.~Teubner,
 Leipzig, 1894.

\bibitem{Cal1}
Calogero F., Solution of a three-body problem in one dimension,
 \href{https://doi.org/10.1063/1.1664820}{\textit{J.~Math. Phys.}} \textbf{10} (1969), 2191--2196.

\bibitem{Cal2}
Calogero F., Solution of the one-dimensional {$N$}-body problems with quadratic
 and/or inversely quadratic pair potentials, \href{https://doi.org/10.1063/1.1665604}{\textit{J.~Math. Phys.}}
 \textbf{12} (1971), 419--436.

\bibitem{CK}
Capel J.J., Kress J.M., Invariant classification of second-order conformally
 flat superintegrable systems, \href{https://doi.org/10.1088/1751-8113/47/49/495202}{\textit{J.~Phys.~A: Math. Theor.}} \textbf{47}
 (2014), 495202, 33~pages, \href{https://arxiv.org/abs/1406.3136}{arXiv:1406.3136}.

\bibitem{EM}
Escobar-Ruiz M.A., Miller Jr. W., Toward a classification of semidegenerate
 3{D} superintegrable systems, \href{https://doi.org/10.1088/1751-8121/aa5843}{\textit{J.~Phys.~A: Math. Theor.}} \textbf{50}
 (2017), 095203, 22~pages, \href{https://arxiv.org/abs/1611.02977}{arXiv:1611.02977}.

\bibitem{Evans}
Evans N.W., Superintegrability in classical mechanics, \href{https://doi.org/10.1103/PhysRevA.41.5666}{\textit{Phys. Rev.~A}}
 \textbf{41} (1990), 5666--5676.

\bibitem{KKM2005c}
Kalnins E.G., Kress J.M., Miller Jr. W., Second-order superintegrable systems
 in conformally flat spaces. {I}.~{T}wo-dimensional classical structure
 theory, \href{https://doi.org/10.1063/1.1897183}{\textit{J.~Math. Phys.}} \textbf{46} (2005), 053509, 28~pages.

\bibitem{KKM2005}
Kalnins E.G., Kress J.M., Miller Jr. W., Second order superintegrable systems
 in conformally flat spaces. {III}.~{T}hree-dimensional classical structure
 theory, \href{https://doi.org/10.1063/1.2037567}{\textit{J.~Math. Phys.}} \textbf{46} (2005), 103507, 28~pages.

\bibitem{KKM2006}
Kalnins E.G., Kress J.M., Miller Jr. W., Second-order superintegrable systems
 in conformally flat spaces. {V}.~{T}wo- and three-dimensional quantum
 systems, \href{https://doi.org/10.1063/1.2337849}{\textit{J.~Math. Phys.}} \textbf{47} (2006), 093501, 25~pages.

\bibitem{KKM2007a}
Kalnins E.G., Kress J.M., Miller Jr. W., Fine structure for 3{D} second-order
 superintegrable systems: three-parameter potentials, \href{https://doi.org/10.1088/1751-8113/40/22/008}{\textit{J.~Phys.~A:
 Math. Theor.}} \textbf{40} (2007), 5875--5892.

\bibitem{KKM2007b}
Kalnins E.G., Kress J.M., Miller Jr. W., Nondegenerate three-dimensional
 complex {E}uclidean superintegrable systems and algebraic varieties,
 \href{https://doi.org/10.1063/1.2817821}{\textit{J.~Math. Phys.}} \textbf{48} (2007), 113518, 26~pages,
 \href{https://arxiv.org/abs/0708.3044}{arXiv:0708.3044}.

\bibitem{KKM2018}
Kalnins E.G., Kress J.M., Miller Jr. W., Separation of variables and
 superintegrability: the symmetry of solvable systems, \textit{IOP Expanding Physics},
 \href{https://doi.org/10.1088/978-0-7503-1314-8}{IOP Publishing}, Bristol, 2018.

\bibitem{K1887}
Koenigs G.X.P., Sur les g\'eod\'esiques a integrales quadratiques, in
 Le{\,c}ons sur la th\'eorie g\'en\'erale des surfaces, Vol.~4, Editor
 J.G.~Darboux, Chelsea Publishing, 1972, 368--404.

\bibitem{Moser}
Moser J., Three integrable {H}amiltonian systems connected with isospectral
 deformations, \href{https://doi.org/10.1016/0001-8708(75)90151-6}{\textit{Adv. Math.}} \textbf{16} (1975), 197--220.

\bibitem{Ran}
Ra\~nada M.F., Superintegrability of the {C}alogero--{M}oser system: constants
 of motion, master symmetries, and time-dependent symmetries, \href{https://doi.org/10.1063/1.532770}{\textit{J.~Math.
 Phys.}} \textbf{40} (1999), 236--247.

\bibitem{RWW}
Rauch-Wojciechowski S., Waksj\"o C., What an effective criterion of
 separability says about the {C}alogero type systems, \href{https://doi.org/10.2991/jnmp.2005.12.s1.43}{\textit{J.~Nonlinear
 Math. Phys.}} \textbf{12} (2005), suppl.~1, 535--547.

\bibitem{SW}
Smirnov R.G., Winternitz P., A class of superintegrable systems of {C}alogero
 type, \href{https://doi.org/10.1063/1.2345472}{\textit{J.~Math. Phys.}} \textbf{48} (2006), 093505, 8~pages,
 \href{https://arxiv.org/abs/math-ph/0606006}{arXiv:math-ph/0606006}.

\bibitem{WS}
Wojciechowski S., Superintegrability of the {C}alogero--{M}oser system,
 \href{https://doi.org/10.1016/0375-9601(83)90018-X}{\textit{Phys. Lett.~A}} \textbf{95} (1983), 279--281.

\end{thebibliography}
\end{document}